\documentclass[11pt,a4paper,reqno]{amsart}

\usepackage{hyperref}
\hypersetup{
  colorlinks   = true, 
  urlcolor     = blue, 
  linkcolor    = Purple, 
  citecolor   = red 
}
\usepackage{changepage}
\usepackage{amsmath}
\usepackage{amsthm}
\usepackage{amsfonts}
\usepackage{amssymb}
\usepackage{array}
\usepackage[margin=2.7cm]{geometry}

\usepackage{soul}
\usepackage{enumitem}
\usepackage{hhline}
\usepackage{multirow} 

\makeatletter
\newcommand{\mylabel}[2]{#2\def\@currentlabel{#2}\label{#1}}
\makeatother

\newcommand{\C}{\mathcal{C}}

\newcommand{\F}{\mathbb{F}}

\newcommand{\N}{\mathbb{N}}

\newcommand{\SN}{\mathcal S_N}

\newcommand{\Fq}{\mathbb{F}_{q}}
\newcommand{\sgn}{\mathrm{sgn}}

\newcommand{\Y}{Y}
\newcommand{\y}{y}

\newcommand{\Gxy}{G\left(x,\Y,\mathrm{B},\Lambda\right)}
\newcommand{\GI}{G_{\bar{I}}\left(x,\Y, \mathrm{B}_{\bar{I}},\Lambda_{\bar{I}}\right)}
\newcommand{\GIbar}{G_{\bar{I}_{0}}\left(x,\Y, \mathrm{B}_{\bar{I}_{0}},\Lambda_{\bar{I}_{0}}\right)}
\newcommand{\GIbarTilde}{\widetilde{G}_{\bar{I}}(x,\Y, \widetilde{\mathrm{B}}_{\bar{I}},\widetilde{\Lambda}_{\bar{I}})}
\newcommand{\wt}{\mathrm{wt}}
\newcommand{\floor}[1]{{\left\lfloor{#1}\right\rfloor}}

\newcommand{\dfree}{\mathrm{d_{free}}}
\newcommand{\dd}{\mathrm{d}}
\newcommand{\cc}{\mathrm{c}}

\newcommand{\GRS}{\mathrm{GRS}}
\newcommand{\ie}{\textit{i.e.}}

\newcommand{\mm}{e}

\DeclareMathOperator{\diag}{diag}

\newtheorem{theorem}{Theorem}[section]

\newtheorem{proposition}[theorem]{Proposition}
\newtheorem{corollary}[theorem]{Corollary}
\newtheorem{lemma}[theorem]{Lemma}

\theoremstyle{definition}
\newtheorem{definition}[theorem]{Definition}
\newtheorem{example}[theorem]{Example}
\newtheorem{remark}[theorem]{Remark}
\newtheorem{notation}[theorem]{Notation}

 \usepackage[dvipsnames,table]{xcolor}
 \usepackage{pagecolor}
 \definecolor{light-gray}{gray}{0.90}

\usepackage{tikz}
\usepackage{lscape}
\usepackage{arydshln}
\newcommand\cellred{\cellcolor{red!20}}
\newcommand\cellblue{\cellcolor{blue!20}}
\newcommand\cellyellow{\cellcolor{orange!20}}
\newcommand\cellgreen{\cellcolor{green!20}}

\title{Weighted Reed-Solomon convolutional codes}

\author[G. N. Alfarano]{Gianira N. Alfarano}
\address{Institute of Mathematics, University of Zurich, Switzerland}
\curraddr{}
\email{gianiranicoletta.alfarano@math.uzh.ch}

\author[D. Napp]{Diego Napp}
\address{Department of Mathematics, University of Alicante, Spain}
\curraddr{}
\email{diego.napp@ua.es}

\author[A. Neri]{Alessandro Neri}
\address{Max-Planck-Institute for Mathematics in the Sciences, Leipzig, Germany}
\curraddr{}
\email{alessandro.neri@mis.mpg.de}

\author[V. Requena]{Ver\'onica Requena}
\address{Department of Mathematics, University of Alicante, Spain}
\curraddr{}
\email{vrequena@ua.es}



\begin{document}

\maketitle
\thispagestyle{empty}
\begin{abstract}
In this paper we present a concrete algebraic construction of a novel class of convolutional codes. These codes are built upon generalized Vandermonde matrices and therefore can be seen as a natural extension of Reed-Solomon  block codes to the context of convolutional codes. For this reason we call them weighted Reed-Solomon  (WRS) convolutional codes. We show that under some constraints on the defining parameters these codes are Maximum Distance Profile (MDP), which means that  they have the maximal possible growth in their column distance profile. We study the size of the field needed to obtain WRS convolutional codes which are MDP and compare it with the existing general constructions of MDP convolutional codes in the literature, showing that in many cases WRS convolutional codes require significantly smaller fields.
\end{abstract}

\bigskip

\section{Introduction}

The algebraic theory of block codes is remarkably elaborated and has produced sophisticated algebraic classes of codes with associated decoding algorithms. On the contrary, there exist very few algebraic general constructions of convolutional codes and most of the existing convolutional codes that have good designed distance have been found by computer search.

Since there is no easy algebraic approach to construct a generator matrix of convolutional codes with good distance properties, several authors have extended  well-known classes of block codes to the convolutional context. This idea was initiated by Massey, Costello and Justesen, who used cyclic or quasi-cyclic block codes \cite{Justesen1972,Justesen1975,Massey1973}. Later, the same idea was further  developed by many authors; see \cite{Esmaeili99,Levy1993,Tanner1987, LAGUARDIA2014}.  The idea of this approach is to establish a link between the generator polynomials of the quasi-cyclic block codes and the generator matrix of convolutional codes. The most important property of this connection is that it allows to lower bound the free distance of the convolutional code by the minimum distance of the associated cyclic or quasi-cyclic block code. Within this setting, many constructions of convolutional codes with designed free distance were provided  based on different classes of block codes, such as Reed--Solomon or Reed--Muller codes. Moreover, in \cite{Smarandache2001}, the authors adjusted the parameters of these constructions to present the first Maximum Distance Separable (MDS) convolutional code, \ie, a convolutional code whose free distance achieves the generalized Singleton bound presented in  \cite{ro99a1}, provided that the field size is congruent to 1 modulo the length of the code. Later, other examples of MDS convolutional codes, also for restricted set of parameters, were presented; see \cite{Gluesing06,Plaza13}. It is worth to mention also the use of circulant Cauchy matrices for the construction of MDS 2D convolutional codes in \cite{cl12,cl16}, that was later adapted for MDS 1D  convolutional codes in \cite{LiebPinto}.  All these codes are designed in such a way that they have large free distance.

In the context of convolutional codes, one aims to build codes that can correct as many errors as possible in different time intervals.  This  property is measured by the notion of \emph{column distances}.  
Despite the fact that this notion  is arguably the most fundamental distance measure for convolutional codes (see \cite[pag. 162]{jo99}), very little is known on how to build convolutional codes with large column distances. Moreover, having large free distance does not guarantee to have the largest possible column distances.
Codes with maximum column distances  are called Maximum Distance Profile (MDP) and they were introduced in \cite{Hutchinson2005} and further investigated in \cite{gl03}. MDP convolutional codes are similar to MDS block codes within a time interval and therefore they are considered a more appropriate analogue of MDS block codes than MDS convolutional codes. From a practical point of view, it has been recently shown that these codes are very appealing for sequential transmission over the erasure channel and low-delay streaming applications; see \cite{ba15b,Khisti19,ma16,to12}.

Surprisingly, there exist only two general algebraic constructions that yield two wide classes of MDP convolutional codes (see \cite{AlmeidaNappPinto2013} and \cite{gl03}), but both require unpractical large field sizes. These classes are built using lower triangular Toeplitz \emph{superregular matrices}, \ie matrices having the property that the minors that are not \textit{trivially zero} are nonzero; see \cite{al16,gl03} for a formal definition and details on the relation between superregular matrices and MDP convolutional codes. Due to the difficulty of deriving general constructions, researchers have been focusing on computer search algorithms for finding MDP convolutional codes with small parameters; see \cite{HaOs2018,AlmeidaLieb,gl03,LiebPinto,to12}.
\bigskip

In this work we use a different approach to derive large classes of MDP convolutional codes and present a new general algebraic construction. Rather than using superregular matrices or generator polynomials of cyclic or quasi-cyclic block codes, we carefully select  different modified Vandermonde matrices as the coefficients of the polynomial generator matrix of the convolutional code in such a way that the resulting  code is, under some constraints, MDP. Since each modified Vandermonde matrix is the generator (or parity-check) matrix of a generalized Reed-Solomon (GRS)  block code, the presented class of codes can be considered as a very natural extension of GRS block codes to the context of convolutional codes. For this reason, we call them \emph{weighted Reed-Solomon (WRS) convolutional codes}. We show that the field size required to build them is significantly smaller than the existing ones in the literature for other classes of MDP codes.

\bigskip

The paper is structured as follows: In Section \ref{sec:prelim} we provide the necessary background on convolutional codes. Section \ref{Sec:Construction} is devoted to present the construction of WRS convolutional codes based on Vandermonde matrices. In Section \ref{sec:proofTheorem} we introduce and study a polynomial generalization of the truncated sliding generator matrix of the constructed convolutional codes. From the obtained results we can show that the class of convolutional codes proposed is MDP, provided a suitable choice of the defining elements. In Section \ref{sec:fieldsize} we investigated the field size required for our WRS convolutional codes to be MDP and we make an asymptotic comparison with existing constructions of MDP convolutional codes in the literature, showing that for many regimes of parameters, our construction requires  smaller fields. We conclude the paper by drawing some conclusive remarks and future works in Section \ref{sec:conclusions}.

\section{Preliminaries}\label{sec:prelim}
In this section we establish the notation that we will use in the remainder of this work and we provide the necessary background on convolutional codes, with a particular focus on MDP ones.

Let $q$ be a prime power and $\F_q$ be the finite field of order $q$. Let $k,n$ be positive integers, with $k\leq n$. An $(n,k)_q$ \textbf{convolutional code} is a submodule $\C$ of $\F_q[z]^n$ of rank $k$ such that there exists a  full row rank \textbf{generator matrix} $G(z)\in \F_q[z]^{k\times n}$ with the property that
$$\C:=\{u(z)G(z)\mid u(z) \in \F_q[z]^k\}\subseteq \F_q[z]^n.$$ 
If $G(z)$ is \textbf{basic}, \ie it has a right polynomial inverse, there exists also a \textbf{parity-check matrix} $H(z)\in\F_q[z]^{(n-k)\times n}$, such that
We define the \textbf{degree} $\delta$ of $\C$ as the highest degree of the $k\times k$  minors of one -- and hence any -- generator matrix of $\C$. We say that $G(z)$ is \textbf{reduced} if the sum of its row degrees is equal to $\delta$, where the $i$-th row degree denotes the highest degree of the polynomials forming the $i$-th row of $G(z)$. When its degree $\delta$ is known, we write that $\C$ is an $(n,k, \delta)_q$ convolutional code.

Due to the natural isomorphism between $\Fq[z]^{k\times n}$ and $\Fq^{k\times n}[z]$,  the encoding map of a convolutional code $\C$ can be expressed via the multiplication by the polynomial $G(z):=\sum_{i=0}^m G_iz^i$, where $G_i\in\Fq^{k\times n}$ for each $i\in\{1,\dots,m\}$, and $G_m\ne 0$. If $G(z)$ is reduced, then the (polynomial) degree $m$ of $G(z)$ is invariant, and it is called \textbf{memory} of the code $\C$.

The  \textbf{weight} of a polynomial $v(z)=\sum_{i=0}^r v_iz^i\in \F_q^n[z]$ is defined as
$$\wt(v(z)) := \sum _{i=0}^r \wt_{\mathrm{H}}(v_i)\in\N_0,$$
where $\wt_{\mathrm{H}}(v_i)$ denotes the Hamming weight of $v_i\in\F_q^n$, \ie, the number of its nonzero components.
Finally, the \textbf{free distance} of a convolutional code $\C$ is the integer
$$\dfree(\C):=\min\{\wt(v(z))\mid v(z)\in\C, v(z)\ne 0\}.$$
The importance of the parameters $\delta$ and $\dfree$ is testified by the fact that they determine respectively the decoding complexity and the error correction capability of the Viterbi decoding algorithm; see \cite{jo99}. For this reason, for any given positive integers $n$ and $k$ and field size $q$, the aim is to construct $(n,k)_q$ convolutional codes with ``small'' degree $\delta$ and  ``large'' free distance $\dfree$. 
The parameters of an $(n,k,\delta)_q$ convolutional code $\C$ are related by the generalized Singleton bound (\cite{ro99a1}):
\begin{equation}\label{eq:genSingleton}
\dfree(\C)\leq (n-k)\left(\floor{\frac{\delta}{k}}+1\right)+\delta +1.
\end{equation}



The difference between convolutional codes and block codes is that convolutional codes treat the information as a sequence of data and the encoding and decoding process at each time instant depends on previous information. One of the main advantages of these codes, in contrast to classic block codes, is that they can consider different sliding windows (time intervals) for decoding, depending on the distribution of errors in the sequence. In this sense, convolutional codes have more flexibility than block codes. Therefore, one aims to build convolutional codes with optimal error correcting properties within windows of different sizes. This property is captured by the notion of column distances. 

Let $v(z) = \sum_{i=0}^r v_iz^i\in\F_q^n[z]$. For any $j\leq r$, let $v_{[0,j]}(z) := \sum_{i=0}^j v_iz^i$.
The \textbf{$j$-th column distance} $\dd_j^\cc$ of the code $\C$ is defined as
$$ \dd_j^\cc(\C) :=\min\{\wt(v_{[0,j]}(z))\mid v(z)\in\C, \ v_0\ne 0\}. $$

\noindent

\noindent
The following result gives a Singleton-like bound  for the column distances of a convolutional code.

\begin{theorem}[\textnormal{\cite[Proposition 2.3]{gl03}}]\label{thm:bound_dj}
Let $\C$ be an $(n,k,\delta)_q$ convolutional code. Then, for any $j\ge 0$,
\begin{equation}\label{eq:col_dist_bound}
\dd_j^\cc(\C)\leq(n-k)(j+1)+1.
\end{equation}
\end{theorem}

Note that the column distances $\dd_j^\cc(\C)$ are invariants of the code $\C$ and they determine the error correction capability of $\C$ to estimate the message symbol $u_0$, based on the received symbols $v_{[0,j]}$; see  \cite{gl03,jo99}. Hence, a good computational performance for sequential decoding requires a rapid initial growth of the column distances.

Since $\dd_j^\cc(\C)\leq \dfree(\C)$ for any $j$, it is a matter of straightforward computations to verify that the maximum instant such that the bound \eqref{eq:col_dist_bound} is achievable is $$L:=\floor{\frac{\delta}{k}} + \floor{\frac{\delta}{n-k}}.$$ The $(L+1)$-tuple of numbers $(\dd_0^\cc(\C),\dots, \dd_L^\cc(\C))$ is called the \textbf{column distance profile} of $\C$.

An $(n,k,\delta)_q$ convolutional code $\C$  whose column distances meet the bound \eqref{eq:col_dist_bound} with equality for all $j \in \{0,\dots,L\}$  is called \textbf{maximum distance profile} (MDP).

Let $G(z)=\sum_{i=0}^mG_iz^i\in\Fq^{k\times n}[z]$, with $G_m\ne 0$ be a generator matrix of an $(n,k)_q$ convolutional code $\C$, then $G(0)=G_0$ is full row rank. For any $j\geq 0$, we define the \textbf{$j$-th truncated sliding generator matrix} as
\begin{equation}\label{eq:slidgen}
    G_j^\cc :=\begin{pmatrix}
G_0 & G_1 & \cdots  & G_j \\
 & G_0 & \cdots & G_{j-1} \\
 & & \ddots & \vdots \\
 & &  & G_0\\
\end{pmatrix}\in\Fq^{(j+1)k\times (j+1)n},
\end{equation}
where $G_j = 0,$ whenever $j>m$; see \cite{jo99}.
 The following is a classical result which serves as a characterization for MDP convolutional codes.

\begin{theorem}\cite[Theorem 2.4]{gl03}\label{thm:characterizationMDP}
Let $G(z) = \sum_{i=0}^m G_iz^i\in\F^{k \times n}[z]$  be a generator matrix of a convolutional code $\C$ with $G(0)=G_0$ full row rank. If  every $(j+1)k \times (j+1)k$ full-size minor of $G_j^\cc$ formed by the columns with indices $1\leq t_1< \dots < t_{(j+1)k}$, where $t_{\ell k+1}>\ell n$ for $\ell=1,\dots,j$, is nonzero, then,   $$\dd_j^\cc(\C)=(n-k)(j+1)+1,$$
\ie, the column distance is optimal at time instant $j$. 
\end{theorem}

For an easier an more compact notation, we introduce the following definition.

\begin{definition}
 Let $G(z) = \sum_{i=0}^m G_iz^i\in\F^{k \times n}[z]$ be a polynomial with $G_m \neq 0$, and let $\delta$ be degree of the convolutional code generated by $G(z)$. We say that $G(z)$ has the \textbf{MDP property} if the $L$-th truncated sliding generator matrix $G_L^\cc$ satisfies condition in Theorem \ref{thm:characterizationMDP}. 
\end{definition}

\begin{remark}
The original assumption in Theorem \ref{thm:characterizationMDP} is that $G(z)$ is basic. However, as it emerges from the original proof, it is not necessary to show that the convolutional code generated by $G(z)$ is MDP. Having $G(z)$ basic, indeed, ensure that the code is noncatastrophic; see also \cite[Remark 2.8]{alfarano2020simplified}. In this way, a similar characterization could be derive from the parity-check matrix. Moreover, notice that in \cite{alfarano2020simplified}, it is shown that in the case $m=\frac{\delta}{k}$, if $G(z) =\sum_{i=1}^m G_iz^i$ has the MDP property, then the convolutional code generated by $G(z)$ is noncatastrophic and one can get rid of the assumption of $G(z)$ being basic.
\end{remark}

\begin{remark}\label{rem:MDPminors}
Here we rephrase the MDP property as follows. Each minor of $G_L^\cc$ obtained by selecting the columns with indices as described in Theorem \ref{thm:characterizationMDP} is a minor obtained by selecting $\ell_i$ columns from the $i$-th columns block, for every $i=0,\ldots, L$, such that
\begin{align}\label{eq:lis}
\begin{split}
 \sum_{i=0}^s \ell_i \leq (s+1)k\,\, & \mbox{ for } s=0,\ldots, L-1, \\
 \sum_{i=0}^L \ell_i = (L+1)k. &
\end{split}
\end{align}

Note that the remaining full size minors of $G_L^\cc$ not satisfying (\ref{eq:lis}) are trivially zero, \ie, are zero independently of the choice of the nonzero entries of $G_L^\cc$, see \cite{al16,gl03} for a formal definition. 
\end{remark}


We conclude this section by recalling the definition of generalized Reed-Solomon codes, which are one of the most studied family of codes in algebraic coding theory, due to their very rich algebraic structure and their  suitability for digital implementation in practical applications: they possess optimal distance and admit efficient algebraic decoding algorithms, \textit{e.g.}, Berlekamp-Massey, see \cite[Chapters 10 and 11]{ma77}. We will use their generator matrices in order to construct MDP convolutional codes. 

Let $0<k\leq n$ be two positive integers and consider the set of polynomials with coefficients in $\Fq$ and degree strictly less than $k$, namely
$$\Fq[x]_{<k} :=\{f(x)\in\Fq[x]\mid \deg f< k\}. $$

\begin{definition}
Suppose that $n\leq q$, and consider $\alpha_1,\dots, \alpha_n \in \Fq$ pairwise distinct elements, and $b_1,\dots, b_n\in\Fq^\ast$. The block code
$$ C :=\{(b_1f(\alpha_1),\dots, b_nf(\alpha_n))\mid f\in\Fq[x]_{<k}\}$$
is called \textbf{generalized Reed-Solomon (GRS) code} and it is denoted by $\GRS_k(\alpha, b)$, where $\alpha:=(\alpha_1,\dots,\alpha_n)$ and $b = (b_1, \dots, b_n)$.
\end{definition}

The canonical generator matrix for a code $C=\GRS_k(\alpha,b)$ has the following form:
$$ G:=\begin{pmatrix}
b_1 & b_2 & \cdots & b_n \\
b_1\alpha_1 & b_2\alpha_2 & \cdots & b_n\alpha_n \\
b_1\alpha_1^2 & b_2\alpha_2^2 & \cdots & b_n\alpha_n^2 \\
\vdots &\vdots & \ddots & \vdots\\
b_1\alpha_1^{k-1} & b_2\alpha_2^{k-1} & \cdots & b_n\alpha_n^{k-1}
\end{pmatrix} = V_k(\alpha)\diag(b),$$
where $V_k(\alpha)$ is a classical Vandermonde matrix of size $k \times n$ of the form
\[
\begin{pmatrix}
1         & 1            & \cdots & 1 \\
\alpha_1  & \alpha_2     & \cdots & \alpha_n \\
\alpha_1^2   & \alpha_2^2   & \cdots & \alpha_n^2 \\
\vdots       &\vdots       &  \vdots  & \vdots\\
\alpha_1^{k-1} &\alpha_2^{k-1} & \cdots & \alpha_n^{k-1}
\end{pmatrix}
\]
 and $\diag(b)$ denotes the diagonal matrix whose diagonal entries are given by $b_1,\dots, b_n$. We call this generator matrix ``canonical'' since it is obtained by evaluating each monomial of the standard $\Fq$-basis of
$\Fq[x]_{<k}$, that is $\{1,x,x^2,\dots,x^{k-1}\}$, in the points $\alpha_1,\dots, \alpha_n$.


\section{New Construction}\label{Sec:Construction}
In this section we present a new algebraic construction of $(n,k,\delta)_q$ MDP convolutional codes with memory $m=\left\lceil\frac{\delta}{k} \right\rceil$. To this end, we use some generalized Vandermonde matrices as the coefficients of the polynomial matrix $G(z)$ describing the code.

\bigskip

Let $k,n$ be positive integers and let $q$ be a prime power, with $k< n<q$ . Let $\alpha :=(\alpha_1,\dots,\alpha_n)\in(\Fq^\ast)^n$, with the $\alpha_i$'s pairwise distinct, and fix  $\gamma$ to be a root of an irreducible polynomial in $\Fq[z]$ of degree $s$, for some suitable integer $s$. Clearly,  $\Fq(\gamma) \cong \F_{q^s}$.

For any $i \geq 0$, set

\begin{equation}\label{GrM}
M_i:=\begin{pmatrix}
\gamma^{\binom{i+1}{2}k - i}\alpha_1^{(i+1)k-1} & \gamma^{\binom{i+1}{2}k - i}\alpha_2^{(i+1)k-1} & \cdots & \gamma^{\binom{i+1}{2}k - i}\alpha_n^{(i+1)k-1} \\
\vdots &\vdots & & \vdots\\
\gamma^{\binom{i}{2}k + i}\alpha_1^{ik+1} &\gamma^{\binom{i}{2}k + i}\alpha_2^{ik+1} & \cdots & \gamma^{\binom{i}{2}k + i}\alpha_n^{ik+1}\\
\gamma^{\binom{i}{2}k}\alpha_1^{ik} & \gamma^{\binom{i}{2}k}\alpha_2^{ik} & \cdots & \gamma^{\binom{i}{2}k}\alpha_n^{ik}
\end{pmatrix},
\end{equation}
 and, for every $i \geq 0$ and $1\leq j \leq k$,
 
\begin{equation}\label{GrN}
N_{i,j}:=\begin{pmatrix}
0 & 0 & \cdots &0 \\
\vdots &\vdots & & \vdots\\
0 & 0 & \cdots &0 \\
\gamma^{\binom{i}{2}k + (j-1)i}\alpha_1^{ik+j-1} & \gamma^{\binom{i}{2}k + (j-1)i}\alpha_2^{ik+j-1} & \cdots & \gamma^{\binom{i}{2}k + (j-1)i}\alpha_n^{ik+j-1} \\
\vdots &\vdots & & \vdots\\
\gamma^{\binom{i}{2}k + i}\alpha_1^{ik+1} &\gamma^{\binom{i}{2}k + i}\alpha_2^{ik+1} & \cdots & \gamma^{\binom{i}{2}k + i}\alpha_n^{ik+1}\\
\gamma^{\binom{i}{2}k}\alpha_1^{ik} & \gamma^{\binom{i}{2}k}\alpha_2^{ik} & \cdots & \gamma^{\binom{i}{2}k}\alpha_n^{ik}
\end{pmatrix}.
\end{equation}

For the binomial coefficients, we use the convention that $\binom{a}{b}=0$ if $a<b$.
Observe that for every $i \geq 0$, it holds $N_{i,k}=M_i$. Moreover, $M_i, N_{i,j} \in \F_{q^s}^{k \times n}$ for every $i\geq 1$, while $M_0, N_{0,j} \in \F_{q}^{k \times n}$. It is easy to see that  the matrix $M_i$ is the generator matrix of  the  $[n,k]_{q^s}$ block code $\mathrm{GRS}_k(\alpha,\alpha^{(ik)})$, where $\alpha^{(ik)}:=(\alpha_1^{ik},\ldots,\alpha_n^{ik})$. 

\begin{definition}\label{def:CWRS}
Let $k,n,\delta$ be positive integers with $0 <k \leq n$, $\alpha=(\alpha_1,\ldots, \alpha_n) \in (\Fq^*)^n$  and $\gamma \in \F_{q^s}$ be as above. Let $m:=\left\lceil \frac{\delta}{k} \right\rceil$ and $t:=\delta-(m-1)k$, and define 
\begin{equation}\label{Gr} G_i:=\begin{cases} M_i & \mbox{ if } 0 \leq i\leq m-1, \\
N_{m,t} & \mbox{ if } i=m.\end{cases} \end{equation}
A convolutional code is called \textbf{weighted Reed-Solomon (WRS) convolutional code} if admits $G(z)=\sum_{i=0}^m G_i z^i$ as generator matrix. We will denote such a code by $\C_{k,n}^\delta(\gamma, \alpha)$. 
\end{definition}

We want to study now the codes $\C_{k,n}^\delta(\gamma, \alpha)$. We start by providing their parameters.

\begin{proposition}\label{prop:parameters}
The code $\C_{k,n}^\delta(\gamma, \alpha)$ is an $(n,k,\delta)_{q^s}$ convolutional code. In particular, the generator matrix $G(z)$ given in Definition \ref{def:CWRS} is reduced.
\end{proposition}
\begin{proof}
Clearly the code $\C_{k,n}^\delta(\gamma, \alpha)$ is defined over $\F_{q^s}$ and has length $n$. Moreover, the matrix $G_0=M_0$ is a Vandermonde matrix and hence it is full rank. This implies that the dimension of the code is $k$. Let $\tilde{\delta}$ be the degree of $\C_{k,n}^\delta$. By definition, the sum of the row degrees  of $G(z)$ is $\delta$, and hence $\tilde{\delta}\leq \delta$. In order to show that $\delta=\tilde{\delta}$, it is enough to show that $G(z)$ is reduced, \ie, the leading row coefficient matrix of $G(z)$, denoted by $G_\infty$, is full row rank; see \cite{forney75} or \cite[Theorem 6.3–13]{ka80}. Recall that $m=\left\lceil \frac{\delta}{k} \right\rceil$ and $t=\delta-(m-1)k$. It is easy to see that the matrix $G_{\infty}$ has the first $k-t$ rows equal to the ones of $G_{m-1}=M_{m-1}$ and the last $t$ rows equal to the ones of $G_m=N_{m,t}$. Such a matrix is a row permutation of
$$ \begin{pmatrix}
\gamma^{\binom{m}{2}k + (t-1)m}\alpha_1^{mk+t-1} & \gamma^{\binom{m}{2}k + (t-1)m}\alpha_2^{mk+t-1} & \cdots & \gamma^{\binom{m}{2}k + (t-1)m}\alpha_n^{mk+t-1} \\
\vdots & \vdots & & \vdots \\
\gamma^{\binom{m}{2}k}\alpha_1^{mk} & \gamma^{\binom{m}{2}k}\alpha_2^{mk} & \cdots & \gamma^{\binom{m}{2}k}\alpha_n^{mk} \\
\gamma^{\binom{m}{2}k - m}\alpha_1^{mk-1} & \gamma^{\binom{m}{2}k - m}\alpha_2^{mk-1} & \cdots & \gamma^{\binom{m}{2}k - m}\alpha_n^{mk-1} \\
\vdots & \vdots & & \vdots \\
\gamma^{\binom{m-1}{2}k + t(m-1)}\alpha_1^{(m-1)k+t} & \gamma^{\binom{m-1}{2}k + t(m-1)}\alpha_2^{(m-1)k+t} & \cdots & \gamma^{\binom{m-1}{2}k + t(m-1)}\alpha_n^{(m-1)k+t} 
\end{pmatrix},$$
which is full rank, since it is a Vandermonde matrix whose rows are multiplied by powers of $\gamma$ and whose columns are multiplied by  $\alpha_i^{(m-1)k+t}$.
\end{proof}

\begin{definition}\label{def:G(x)}
Let $k,n$ and $\delta$ be fixed. Consider the matrices $G_i(x)$ as the matrices  $G_i$  defined in \eqref{Gr} where we have replaced $\gamma$ by an indeterminate $x$, and let $G_L^\cc(x)$ be the corresponding $L$-th truncated sliding generator matrix. We define the set
\begin{align*}
\mathcal P(k,n,\delta, \alpha):= & \left\{ p(x) \in \Fq[x] \mid p(x) \mbox{ is a full size minor of } G_L^\cc(x)  \mbox{ obtained selecting the columns } \right.\\
& \mbox{ with indices } \left.1\leq j_1< \dots < j_{(L+1)k}, \mbox{ where } j_{rk+1}>rn \mbox{ for } r=1,\dots,L \right\}.
\end{align*}
\end{definition}

Note that the set $\mathcal P(k,n,\delta, \alpha)$ represents the full size minors in $G_L^\cc(x)$ formed as stated in Theorem \ref{thm:characterizationMDP}.

\begin{theorem}\label{lem:MDPpolynomial}
Let $\C_{k,n}^{\delta}(\gamma, \alpha)$ be the $(n,k,\delta)_{q^s}$ WRS convolutional code with  generator matrix $G(z) = \sum_{i=0}^m G_iz^i\in\F_{q^s}^{k\times n}[z]$, where the $G_i$'s are defined by  \eqref{Gr}. If $p(\gamma)\neq 0$ for every $p(x) \in \mathcal P(k,n,\delta, \alpha)$, then $\C_{k,n}^{\delta}(\gamma, \alpha)$ is an MDP convolutional code.
\end{theorem}

\begin{proof}
As $p(\gamma)\neq 0$ for every $p(x) \in \mathcal P(k,n,\delta, \alpha)$, then the condition in Theorem \ref{thm:characterizationMDP} is satisfied for $j=L$ and $G(z)$ has the MDP property, \ie,  $\dd_L^\cc(\C_{k,n}^{\delta}(\gamma, \alpha))=(n-k)(L+1)+1$. It follows from \cite[Corollary 2.3]{gl03} that  $\dd_j^\cc(\C_{k,n}^{\delta}(\gamma, \alpha))=(n-k)(j+1)+1$ for $j=0,1,\dots, L$ and therefore, by definition, $\C_{k,n}^{\delta}(\gamma, \alpha)$ is an $(n,k,\delta)$ MDP convolutional code. 
\end{proof}

For a given nonzero polynomial $p(x) \in \Fq[x]$, we denote by $\deg p(x)$ the degree of $p(x)$, and by $\nu (p(x))$ the maximum integer $\ell$ such that $x^\ell$ divides $p(x)$. Then we define the integer
$$D(k,n,\delta,\alpha):=\max\{ \deg p(x)-\nu (p(x)) \mid 0 \neq  p(x) \in \mathcal P(k,n,\delta,\alpha) \}.$$

The next result is the main theorem of this section. However, its proof requires several technical lemmas and it  can be found in  Section \ref{sec:proofTheorem}.

\begin{theorem}\label{thm:ConstructionMDP}
Let $\gamma$ be a root of an irreducible polynomial in $\Fq[z]$ of degree $s$ and let  $\C_{k,n}^{\delta}(\gamma, \alpha)$ be the $(n,k,\delta)_{q^s}$ WRS convolutional code whose generator matrix is $G(z) = \sum_{i=0}^m G_iz^i\in\F_{q^s}^{k\times n}[z]$, and the $G_i$'s are defined by \eqref{Gr}. If $s>D(k,n,\delta,\alpha)$, then $\C_{k,n}^\delta(\gamma, \alpha)$ is an MDP convolutional code in $\F_{q^s}[z]^n$.
\end{theorem}

We conclude this section by illustrating with a concrete example how to construct a WRS convolutional code that is also MDP, using the previous theorem. 

\begin{example} \label{Ex:WRSvalues}
We fix the parameters $k=3$, $n=5$ and  $\delta=5$.  Therefore, we have $m=2$ and  $L=3$. We then choose a prime power greater than $n$, that is $q=7$ and a vector with pairwise distinct nonzero entries  $\alpha=(\alpha_1, \alpha_2, \alpha_3, \alpha_4, \alpha_5)=(1, 2, 3, 4, 5)\in \F_7^5$. At this point we illustrate how to choose a suitable $\gamma$ so that the resulting code $\C_{3,5}^5(\gamma, \alpha)$ is MDP. We consider the polynomial version of the $3$-th truncated sliding generator matrix $G_3^\cc(x)$, given by 

\begin{align*}
G_3^\cc(x) & =
\begin{pmatrix}
G_0(x) & G_1(x) & G_2(x)  &      \\
    & G_0(x) & G_1(x) & G_2(x)  \\
    &     & G_0(x) & G_1(x)  \\
    &     &     & G_0(x) \\
\end{pmatrix}\in \F_7[x]^{12\times 20},
\end{align*}
where
\begin{align*}
G_0(x) & =
\begin{pmatrix}
1 & 4 & 2 & 2 & 4 \\
1 & 2 & 3 & 4 & 5 \\
1       &     1 	& 1			&   1     &     1
\end{pmatrix} \in \F_7^{3\times 5},\\
G_1(x) & =
\begin{pmatrix}
x^2 & 4x^2 &  5 x^2 &  2x^2 &  3 x^2 \\
x  &  2 x &  4 x &  4x & 2 x \\
1 & 1 & 6 & 1 & 6
\end{pmatrix}\in \F_7[x]^{3\times 5},\\
G_2(x) & =
\begin{pmatrix}
0 & 0 & 0 & 0 & 0 \\
x^5  &  2 x^5 &  3x^5 &  4x^5 &  5x^5 \\
x^3  & x^3  & x^3  & x^3  & x^3 
\end{pmatrix}\in \F_7[x]^{3\times 5}.
\end{align*}

We now compute the value $D(3,5,5,\alpha)$, which can be checked to be  $D(3,5,5,\alpha)=9$. There are many full size minors of $G_{3}^\cc(x)$ from which we can obtain this value. For instance, if we select the columns with indices $\{1,6,7,8,11,12,13,16,17,18,19,20\}$ of $G_{3}^\cc(x)$,
we have that the full size minor is  $p(x)=x^3(3x^9+2x^8+4x^7+5x^5+x^4+4x^3+x^2+4x+4)  \in \mathcal P(3,5,5,\alpha)$. Let now choose $\gamma$ to be a root of an irreducible polynomial of degree $s=10$ over $\F_7$. Thus, with this choice, the code $\C_{3,5}^5(\gamma,\alpha)$ is an MDP $(5,3,5)$ WRS convolutional code over the field $\F_{7^{10}}$. Its generator matrix is given by $G(z):=G_0(\gamma)+G_1(\gamma)z+G_2(\gamma)z^2$. 
\end{example}

\section{A multivariate polynomial generalization of $G_L^\cc$ }\label{sec:proofTheorem}

In Section \ref{Sec:Construction} we introduced $G_L^\cc(x)$ in Definition \ref{def:G(x)} as a polynomial generalization of the truncated sliding generator matrix $G_L^\cc$ of WRS convolutional codes, by substituting $\gamma$ with a variable $x$. In this section we further generalize  its square submatrices  by seeing the $\alpha_i$'s defining the generalized Vandermonde matrices as algebraically independent variables $y_i$'s, yelding a multivariate polynomial representation of $G_L^\cc$. This generalization allows to give a proof of Theorem \ref{thm:ConstructionMDP}, to which this section is dedicated.
For the convenience of the reader, in this short introduction we briefly present the idea of the proof. 

We denote by $\mathrm{B}$ the collection of the involved powers of $x$, and by $\Lambda$ the collection of the exponents of $y_i$'s involved in the generalized Vandermonde matrices constituting the matrix $G_L^\cc(x)$. In other words, $\mathrm{B}$ and $\Lambda$ denote the exponents of the variables. In this way one obtains a polynomial generalization of the square submatrices of $G_L^\cc$, denoted by $G(x,\Y, \mathrm{B}, \Lambda)$, and their  minors become multivariate polynomials $p(x, \Y)$, where $\Y$ denotes the vector formed by the variables $y_i$'s. 

Several technical lemmas lead to Theorem \ref{thm:main}, where we describe the monomial of minimal degree of $p(x, \Y)$ in the variable $x$, which is given by the product of determinants of some particular submatrices of $G(x,\Y, \mathrm{B}, \Lambda)$. 

By choosing some special values of $\mathrm{B}$ and $\Lambda$ and specializing $\Y$ in a suitable vector $\mathrm{A}$ of elements in $\Fq$, we obtain that the resulting matrix yields a square submatrix of $G_L^\cc(x)$ as in Definition \ref{def:G(x)}. Moreover, we show  that the monomial of minimal degree in $p(x, \Y)$ is still nonzero when $\Y$ is specialized in $\mathrm{A}$. In particular, the set $\mathcal P(k,n,\delta,\alpha)$ defined in Definition \ref{def:G(x)} will consists only of such polynomials $p(x,\mathrm{A})$, which are all nonzero. By carefully choosing the value $\gamma$, we then show that the resulting convolutional code $\C_{k,n}^\delta(\gamma,\alpha)$ is MDP, by means of Theorem \ref{lem:MDPpolynomial}.
In Theorem \ref{thm:maxdegree} we give the equivalent version of Theorem \ref{thm:main} for the monomial of maximum degree in $x$ of the same polynomial $p(x,\Y)$. However, when considering $p(x,\mathrm{A})$, such  monomial could  vanish. 

All the results mentioned above are needed to finally prove Theorem \ref{thm:ConstructionMDP}, which states that WRS convolutional codes are MDP.

We start by recalling the definition of generalized Vandermonde matrix. Then, we establish the notation for the remainder of the section.

\begin{definition}\label{def:GenVandermonde}
Let $\Fq$ be the finite field with $q$ elements, $k,n$ be  positive integers. Let $\lambda = (\lambda_1,\dots, \lambda_k)\in \N^k$ be a vector whose entries are pairwise distinct and $\alpha = (\alpha_1,\dots, \alpha_n)\in \Fq^n$.
A $k\times n$ \textbf{generalized Vandermonde matrix} is a matrix of the form 
$$V({\lambda},\alpha) = \begin{pmatrix}
\alpha_1^{\lambda_1} & \alpha_2^{\lambda_1} & \cdots & \alpha_n^{\lambda_1} \\
\alpha_1^{\lambda_2} & \alpha_2^{\lambda_2} & \cdots & \alpha_n^{\lambda_2} \\
\vdots & \vdots & \ddots & \vdots\\
\alpha_1^{\lambda_k} & \alpha_2^{\lambda_k} & \cdots & \alpha_n^{\lambda_k} \\
\end{pmatrix}\in \Fq^{k\times n}. $$

\end{definition}

The following definition introduces the polynomial matrix which is central to this section. 

\begin{definition}\label{def:shapeG}
Let $\mm\in \N$ be a nonegative integer, 
 $(\ell_0, \dots, \ell_\mm), (k_0,\dots, k_\mm) \in (\N_{>0})^{\mm+1}$, such that $\sum_{i=0}^r \ell_i \leq \sum_{i=0}^r k_i$ for any $r\in\{0,\dots,\mm-1\}$ and $\sum_{i=0}^\mm \ell_i = \sum_{i=0}^\mm k_i$.
For any $j\in\{0,\dots, \mm\}$, let $\y^{(j)} = \left(y_1^{(j)},\dots, y_{\ell_j}^{(j)}\right)$ be a vector of variables and for any $0\leq i \leq j\leq \mm$, let $\lambda^{(i,j)}=(\lambda^{(i,j)}_1, \dots, \lambda^{(i,j)}_{k_i}) \in \N^{k_i}$ be such that the following conditions hold:

\begin{enumerate}[label=(\roman*)]
\item[\mylabel{lambdagrows}{\textnormal{(L1)}}]  $\lambda_{s-1}^{(i,j)} > \lambda_s^{(i,j)}$, for any $s\in\{2,\dots,k_i\}$.
\item[\mylabel{lambda_ii}{\textnormal{(L2)}}] $\lambda_{1}^{(i,j)} > \lambda_{k_{i+1}}^{(i+1,j)}$, for any $0\leq i\leq j-1$, $1\leq j\leq \mm$.
\end{enumerate}
For any $0\leq i \leq j\leq \mm$, let $\beta^{(i,j)}=(\beta^{(i,j)}_1,\dots, \beta^{(i,j)}_{k_i})\in\N^{k_i}$, such that:

\begin{enumerate}[label=(\alph*)]
\item[\mylabel{part_a}{\textnormal{(b1)}}] $\beta^{(i,i)}=0$.

\item[\mylabel{part_b}{\textnormal{(b2)}}] $\beta^{(i,j)}_{s-2} - \beta^{(i,j)}_{s-1}\geq\beta^{(i,j)}_{s-1}-\beta^{(i,j)}_{s}$, for any $s\in\{3,\dots,k_i\}$.

\item[\mylabel{part_c}{\textnormal{(b3)}}] $\beta^{(i,j)}_{k_i-1} - \beta^{(i,j)}_{k_i} \geq \beta^{(i,j)}_{k_i} - \beta^{(i+1,j)}_1 + 1\geq \beta^{(i+1,j)}_1 - \beta^{(i+1,j)}_2 +1 $, for any $0\leq i \leq j-1$ and $1 \leq j\leq \mm$.

\item[\mylabel{part_d}{\textnormal{(b4)}}]  $\left(\beta^{(i,j+1)}_{s-2} - \beta^{(i,j)}_{s-2}\right) - \left(\beta^{(i,j+1)}_{s-1} - \beta^{(i,j)}_{s-1}\right) \geq \left(\beta^{(i,j+1)}_{s-1} - \beta^{(i,j)}_{s-1}\right) - \left(\beta^{(i,j+1)}_{s} - \beta^{(i,j)}_{s}\right)$ for any $s\in\{3,\dots,k_i\}$, for any $i\leq j$ and $0\leq j \leq \mm-1$.

\item[\mylabel{part_e}{\textnormal{(b5)}}] $\left(\beta^{(i,j+1)}_{k_i} - \beta^{(i,j)}_{k_i}\right) - \beta^{(i+1,j+1)}_{1} \geq  \beta^{(i+1,j+1)}_{1} -  \beta^{(i+1,j+1)}_{2}$, for any $0\leq i\leq j \leq \mm-1$.
\end{enumerate}
We define 
\begin{equation}\label{eq:Aij}
    A_{i,j}^{\left(\beta^{(i,j)},\lambda^{(i,j)}\right) }:= \mathrm{diag}\left(x^{\beta^{(i,j)}}\right)V\left(\lambda^{(i,j)},\y^{(j)}\right)\in \F_q[x,\y^{(j)}]^{k_i\times\ell_j},
\end{equation}
where 
$$ \mathrm{diag}\left(x^{\beta^{(i,j)}}\right) = \begin{pmatrix}
x^{\beta^{(i,j)}_1} & 0 & \cdots & 0 \\
0 & x^{\beta^{(i,j)}_2} & \cdots & 0 \\
\vdots & \vdots & \ddots & \vdots \\
0 & 0 & \cdots & x^{\beta^{(i,j)}_{k_i}}
\end{pmatrix}\in\F[x]^{k_i\times k_i}.$$
 To simplify the notation in \eqref{eq:Aij}, we only write $A_{i,j}$  and specify the vectors $,\lambda^{(i,j)}$,  $\beta^{(i,j)}$ only when it is necessary. 
Let $$\Y:=\left(\y^{(0)}_1,\dots,\y^{(0)}_{\ell_0},\y^{(1)}_1,\dots,\y^{(1)}_{\ell_1},\dots, \y^{(\mm)}_1,\dots,\y^{(\mm)}_{\ell_\mm}\right)$$ be the vector of all the variables and
$$\mathrm{B} := \left(\beta^{(0,0)}, \dots, \beta^{(0,e)}, \beta^{(1,1)},\dots, \beta^{(1,e)}, \dots,  \beta^{(e,e)}\right),$$
$$\Lambda := \left(\lambda^{(0,0)}, \dots, \lambda^{(0,e)}, \lambda^{(1,1)},\dots, \lambda^{(1,e)}, \dots, \lambda^{(e,e)}\right)$$ be vectors of exponents. These three vectors uniquely determine the following matrix
\begin{equation}\label{matrixshapeG}
G(x,\Y, \mathrm{B}, \Lambda) := \begin{pmatrix}
A_{0,0} & A_{0,1} & \cdots & A_{0,\mm} \\
 & A_{1,1} &\cdots & A_{1,\mm} \\
 & & \ddots & \vdots \\
& & & A_{\mm,\mm}
\end{pmatrix} \in\F_q[x,\Y]^{(k_0 + \dots + k_\mm) \times (\ell_0 + \dots +\ell_\mm)}.
\end{equation}
\end{definition}
\bigskip

In the next example, we provide some tuples satisfying conditions \ref{lambdagrows}--\ref{lambda_ii} and \ref{part_a}--\ref{part_e}, to get a more intuitive idea of their relations.
\begin{example}\label{ex:exponents}
Let $\mm=2$ and  $(k_0,k_1,k_2)=(4,4,4), (\ell_0,\ell_1,\ell_2)=(2,4,6) \in \N^3$.
Let
$$
\begin{array}{rclrclrcl}
\lambda^{(0,0)}& = & (3,2,1,0), &
\lambda^{(0,1)}& = & (7,6,5,4), &
\lambda^{(0,2)}& = & (11,10,9,8) \\
& & & 
\lambda^{(1,1)}& = & (3,2,1,0), &
\lambda^{(1,2)}& = & (7,6,5,4), \\
& & & & & &
\lambda^{(2,2)}& = & (3,2,1,0).\\
\end{array}
$$
For convenience, we ordered these vectors in row/column blocks. Each row block corresponds to a fixed $i$, and each column block corresponds to a fixed $j$. Clearly, all the vectors defined above satisfy condition \ref{lambdagrows}, \ie are ordered in a decreasing order. Property \ref{lambda_ii} is referred to vectors in consecutive row blocks but same column block. It states that for any $i\leq j-1\leq 1$, we require the first entry of $\lambda^{(i,j)}$  to be strictly greater than the last entry of $\lambda^{(i+1,j)}$. In this example, we only need to  check the column block defined by $j=1$ and immediately obtain that  $\lambda^{(0,1)}_1 >\lambda^{(1,1)}_4$.

Let
$$
\begin{array}{rclrclrcl}
\beta^{(0,0)}& = & (0,0,0,0), &  \beta^{(0,1)}& = & (3,2,1,0), &
\beta^{(0,2)}& = & (10,8,6,4), \\
& & &
\beta^{(1,1)}& = & (0,0,0,0), &
\beta^{(1,2)}& = & (3,2,1,0),\\
& & & & & &
\beta^{(2,2)}& = & (0,0,0,0).\\
\end{array}
$$
Condition \ref{part_a} is clearly satisfied. Condition \ref{part_b} refers to each vector $\beta^{(i,j)}$. It states that the differences between consecutive entries are non increasing. For instance, consider the vector $\beta^{(0,2)}$. We have $$\beta^{(0,2)}_1-\beta^{(0,2)}_2\geq \beta^{(0,2)}_2-\beta^{(0,2)}_3 \geq \beta^{(0,2)}_3-\beta^{(0,2)}_4.$$
Condition \ref{part_c} refers to two vectors of two different row blocks, but same column block, for instance, $\beta^{(0,2)}$ and $\beta^{(1,2)}$. In this example, this is the only possible pair on which this property can be verified. We have
\begin{align*}
    \beta^{(0,2)}_3-\beta^{(0,2)}_4\geq\beta^{(0,2)}_4-\beta^{(1,2)}_1+1\geq\beta^{(1,2)}_1-\beta^{(1,2)}_2+1
\end{align*} 
which is obviously satisfied since $6-4\geq 4-3+1 \geq 3-2+1$. 
Condition \ref{part_d} refers to two vectors of the same row block. Consider $\beta^{(0,0)}$ and $\beta^{(0,1)}$. We want that 
$$  (\beta^{(0,1)}_1-\beta^{(0,0)}_1) - (\beta^{(0,1)}_2-\beta^{(0,0)}_2) \geq (\beta^{(0,1)}_2-\beta^{(0,0)}_2) - (\beta^{(0,1)}_3-\beta^{(0,0)}_3).$$ 
Indeed, we have $(3-0)-(2-0)\geq (2-0)-(1-0),$
which is trivially true.
Finally, condition \ref{part_e} relates two vectors of one row block with one vector of the consecutive row block. Consider $\beta^{(0,1)}$, $\beta^{(0,2)}$ and $\beta^{(1,2)}$. We want $$ \beta^{(0,2)}_4 - \beta^{(0,1)}_4 - \beta_{1}^{(1,2)}\geq \beta_1^{(1,2)} -\beta_2^{(1,2)}, $$
which is, also in this case, trivially satisfied since $4-0-3\geq 3-2$. 
\end{example}
Note that in Example \ref{ex:exponents}, all the inequalities are in fact equalities. This is due on purpose, since this particular case will lead to the construction of a WRS convolutional code.

\bigskip

Next, we illustrate with another example the link between the matrices $\Gxy$ -- together with their parameters -- and our family of WRS convolutional codes. Indeed, if we take a WRS convolutional code $\C_{k,n}^\delta(\gamma,\alpha)$ whose parameters satisfy certain conditions, then some of the full-size submatrices of $G_L^\cc(x)$ are obtained from $\Gxy$  after carefully choosing the vectors of exponents $\mathrm{B}$ and $\Lambda$ introduced in Definition \ref{def:shapeG}, and specializing the vector $Y$ in a suitable vector of elements $\alpha_{j_i}$'s obtained from $\alpha$.
Observe that the following example only illustrates a special case and it is meant to guide the reader in understanding our approach. The general case is analyzed later, in the proof of Theorem \ref{thm:ConstructionMDP}.


\begin{example} 
Let $(\ell_0, \dots, \ell_\mm ),(k_0,\dots, k_\mm)\in \mathbb{N}^{\mm+1}$ be vectors as in Definition \ref{def:shapeG}, with $k_i=k\in \mathbb{N}$, for $i=0,1,\ldots,\mm$. Define the vectors $\mathrm{B}=\left(\beta^{(0,0)}, \dots, \beta^{(0,e)}, \beta^{(1,1)},\dots, \beta^{(1,e)}, \dots,  \beta^{(e,e)}\right)$ and $\Lambda := \left(\lambda^{(0,0)}, \dots, \lambda^{(0,e)}, \lambda^{(1,1)},\dots, \lambda^{(1,e)}, \dots, \lambda^{(e,e)}\right)$ as 
$$
 \begin{array}{ll}
   \beta^{(i,j)}& =
     \left( \binom{j-i+1}{2}k-(j-i),\ldots,\binom{j-i}{2}k+(j-i),\binom{j-i}{2}k \right), \\
 \lambda^{(i,j)} &=
     \left( (j-i+1)k-1,\ldots,(j-i)k+1,(j-i)k \right),  
 \end{array}
 $$
for each $i,j$ such that $0\leq j-i\leq \mm$. Let $\Y:=(\y^{(0)}_1,\dots,\y^{(0)}_{\ell_0},\y^{(1)}_1,\dots,\y^{(1)}_{\ell_1},\dots, \y^{(\mm)}_1,\dots,\y^{(\mm)}_{\ell_\mm})$ be the vector of variables and consider the matrix $\Gxy\in \Fq[x,\Y]^{(\mm+1)k\times (\mm+1)k}$.

Now, choose a WRS convolutional code $\mathcal{C}_{k,n}^{\delta}(\gamma,\alpha)$ for suitable $\gamma$ and $\alpha$, where we make two assumptions on the parameters. We select $n,k$ and $\delta$ such that  $\delta=k\mm$ and $k\mm<n-k$. The latter assumption implies $L=\mm$. Consider the generator matrix of $\C_{k,n}^{\delta}(\gamma,\alpha)$ to be $G(z)=\sum_{i=0}^\mm G_iz^i$ as in Definition \ref{def:CWRS}, and take the polynomial version of its $L$-th truncated sliding generator matrix 
\begin{align*}
G_L^\cc(x) & =
\begin{pmatrix}
G_0(x) & G_1(x)  &      \ldots & G_\mm(x)\\
    & G_0(x)  & \ldots &  G_{\mm-1}(x)\\
    &       & \ddots & \vdots \\
    &     &     &  G_0(x) \\
\end{pmatrix}\in \F_{q}[x]^{(\mm+1)k \times (\mm+1)n}.
\end{align*}

We now point out that every full size submatrix of $G_L^\cc(x)$ obtained by taking $\ell_i$ columns from the $i$-th column block, can be derived starting from $\Gxy$ in the following way. Let $J_0,\ldots J_\mm \subseteq \{1,\dots, n\}$ be the corresponding indices of columns that are selected in each block, with $|J_i|=\ell_i$. Then, specialize $y^{(i)}=(y^{(i)}_1,\ldots, y^{(i)}_{\ell_i})$ in the elements $\alpha^{(i)}:=(\alpha_j : j\in J_i)$. Denote $\mathrm{A}:=(\alpha^{(0)},\ldots \alpha^{(\mm)})$ and observe now that, by construction, the selected submatrix of $G_L^\cc(x)$ coincides with  $G(x,\mathrm{A},\mathrm{B},\Lambda)$.

Finally, notice that the conditions on the $\ell_i$'s given in Definition \ref{def:shapeG} coincides with the conditions for the MDP property in \eqref{eq:lis}. However, in this case it is not contemplated that the values $\ell_i$'s can also be zero. We will study how to obtain the full size submatrices of $G_L^\cc(x)$ for this general case in the proof of Theorem \ref{thm:ConstructionMDP}.

\end{example} 

\bigskip

From Definition~\ref{def:shapeG}, we have that $\ell_0\leq k_0$, so we can assume that there exists $r\geq 0$ such that $k_0 = \ell_0 + r$. Moreover, from condition \ref{part_a} it is clear that for any $0\leq i \leq \mm $,  $A_{i,i}\in\F_q[\Y]^{k_i\times\ell_i}$, hence it does not depend on $x$.

We fix some further notation.
\begin{notation}
For any positive integer $i$, we denote $[i] := \{1,\dots , i\}$.
For any $I\subseteq[\ell_0+r]$, such that $|I| = \ell_0$,  we denote by $A_I(\Y)$ the $\ell_0 \times \ell_0$ matrix obtained from $A_{0,0}$ by selecting the rows indexed by $I$. We denote by $I_0$ the set of indices $[\ell_0]$.

Moreover, for any $I\subseteq[\ell_0+r]$ we denote by $\bar{I}$ the complement of $I$ in $[\ell_0 + r]$ and by $\GI$ the $((k_0-\ell_0)+k_1+\dots + k_\mm )\times (\ell_1+\dots+\ell_\mm )$ submatrix of  $\Gxy$ obtained by erasing the first $\ell_0$ columns and the rows indexed by $I$. This deletion automatically determines two new collections of vectors $\mathrm{B}_{\bar{I}}$ and $\Lambda_{\bar{I}}$.
We denote by $\beta^{(0,j)}_{\bar{I}}$ and by $\lambda^{(0,j)}_{\bar{I}}$ the vectors obtained from $\beta^{(0,j)}$ and $\lambda^{(0,j)}$ respectively, after deleting the entries indexed by $I$. 
Finally, we set  $$b_{\bar{I}} = \sum_{s\not\in I} \beta_s^{(0,1)}\in\N.$$
\end{notation}

\begin{remark}\label{rem:deletion_component}
In $\mathrm{B}_{\bar{I}}= (\beta^{(i,j)}_{\bar{I}})_{i,j}$, the deletion of the components indexed by $I$ only regards $\beta^{(0,j)}$. 
\end{remark}


The results presented in the remainder of this section refer all to a matrix $\Gxy$ of the form given in Definition \ref{def:shapeG}, with component matrices  $A^{\left(\beta^{(i,j)},\lambda^{(i,j)}\right)}$, $0\leq i,j\leq \mm $, where $\beta^{(i,j)},\lambda^{(i,j)}$ satisfy conditions \ref{lambdagrows}--\ref{lambda_ii}, \ref{part_a}--\ref{part_e} for any $i,j$.  

\begin{lemma}\label{lemma:shapeGtilde}
With the notation above, the following hold:

\begin{enumerate}
\item If $\ell_0 =  k_0$, then the vectors forming  $\mathrm{B}_{{\bar{I}}_0}$ and $\Lambda_{{\bar{I}}_0}$ -- which define $\GIbar$ -- satisfy conditions \ref{lambdagrows}--\ref{lambda_ii}, \ref{part_a}--\ref{part_e}. 
\item If $\ell_0< k_0$, then, for any $I\subseteq[\ell_0+r]$ of cardinality $\ell_0$, the matrix $\GI$ can be written as
\[\begin{pmatrix}
\diag\left(x^{\beta_{\bar{I}}^{(0,1)}}\right) & 0 \\
0 & \mathrm{Id}
\end{pmatrix} \GIbarTilde,
 \]
where $\mathrm{Id}$ is the identity matrix and $\GIbarTilde$ is a $({k_0'}+k_2+\dots + k_\mm ) \times (\ell_1 + \dots +\ell_\mm )$ matrix of the form \eqref{matrixshapeG}, where $k_0'=k_0+k_1-\ell_0$ and whose defining vectors in $\widetilde{\mathrm{B}}_{\bar{I}}$ and $\widetilde{\Lambda}_{\bar{I}}$ satisfy conditions \ref{lambdagrows}--\ref{lambda_ii}, \ref{part_a}--\ref{part_e}. 
\end{enumerate}
\end{lemma}

\begin{proof}
\begin{enumerate}
\item If $\ell_0 = k_0$, then
\[ 
\GIbar =  \begin{pmatrix}
 A_{0,0}^\prime & A_{0,1}^\prime & \cdots & A_{0,\mm -1}^\prime \\
& A_{1,1}^\prime &\cdots & A_{1,\mm -1}^\prime \\
 & & \ddots & \vdots \\
& & & A_{\mm -1,\mm -1}^\prime
\end{pmatrix},
\]
where $A_{i,j}^\prime = A_{i+1,j+1}$ for any $0\leq i,j\leq \mm -1$, therefore the vectors $\beta^{(i,j)}, \lambda^{(i,j)}$ defining each $A_{i,j}$ clearly satisfy conditions \ref{lambdagrows}--\ref{lambda_ii}, \ref{part_a}--\ref{part_e}.

\item Assume $\ell_0< k_0$. Note that, for any set of indices $I\subseteq [\ell_0 + r]$,  $\GI$ is a $({k_0'}+k_2+\dots + k_\mm ) \times (\ell_1 + \dots +\ell_\mm )$, where $k_0^\prime = k_0+k_1-\ell_0$.
Clearly, 
\[\GI = \begin{pmatrix}
\diag\left(x^{\beta_{\bar{I}}^{(0,1)}}\right) & 0 \\
0 & \mathrm{Id}
\end{pmatrix}\GIbarTilde,\]
and 
\[ 
\GIbarTilde =  \begin{pmatrix}
 A_{0,0}^\prime & A_{0,1}^\prime & \cdots & A_{0,\mm -1}^\prime \\
& A_{1,1}^\prime &\cdots & A_{1,\mm -1}^\prime \\
 & & \ddots & \vdots \\
& & & A_{\mm -1,\mm -1}^\prime

\end{pmatrix},
\]
where  $A_{i,j}^\prime = A_{i+1,j+1}$ for any $1\leq i,j\leq \mm -1$.
Hence, for $i\geq 1$ all the conditions are satisfied. It is left to prove the result for $i=0$. 

Conditions \ref{lambdagrows}--\ref{lambda_ii} are trivially satisfied, since they are related to the vectors of exponents in the generalized Vandermonde matrices, on which we do not make operations.

If $\widetilde{\mathrm{B}}_{\bar{I}} = (\widetilde{\beta}^{(i,j)})_{i,j}$, then, because of Remark \ref{rem:deletion_component}, it follows that for any $0\leq j \leq \mm -1$, we have that
\begin{equation}\label{eq:newexponent}
\widetilde{\beta}^{(0,j)}=\left(\beta_{\bar{I}}^{(0,j+1)}|\beta^{(1,j+1)}\right) - \left(\beta_{\bar{I}}^{(0,1)}|0\right),\end{equation}
where $0$ represents the zero vector and the difference is made componentwise. Here, we used the concatenation symbol just to stress that the deletion of components indexed by $I$ only regards the $0$-th row.

\begin{enumerate}[label=(\alph*)]
\item[{\textnormal{(b1)}}]  $\widetilde{\beta}^{(0,0)}={0}$.

\item[{\textnormal{(b2)}}]\label{InsideBlock} To show this, we have to consider three cases:
	\begin{enumerate}
	\item If $s\in\{3,\dots,k_0-\ell_0-2\}$ we consider $\widetilde{\beta}^{(0,j)}=\beta_{\bar{I}}^{(0,j+1)} - \beta_{\bar{I}}^{(0,1)}$, with abuse of notation, in the sense of \eqref{eq:newexponent}. For any $1\leq h\leq k_0-\ell_0$, we denote by $\beta_{\bar{I},{h}}^{(0,j)}$ the $h$-th entry of $\beta_{\bar{I}}^{(0,j)}$. We need to verify that
$$\beta_{\bar{I},{s-2}}^{(0,j+1)} - \beta_{\bar{I},s-2}^{(0,1)} - \beta_{\bar{I},{s-1}}^{(0,j+1)} + \beta_{\bar{I},s-1}^{(0,1)} \geq \beta_{\bar{I},{s-1}}^{(0,j+1)} - \beta_{\bar{I},s-1}^{(0,1)} - \beta_{\bar{I},{s}}^{(0,j+1)} + \beta_{\bar{I},s}^{(0,1)}.$$

This is true, since by assumption $$\beta_{\bar{I},s-2}^{(0,1)} - \beta_{\bar{I},s-1}^{(0,1)} \geq  \beta_{\bar{I},s-1}^{(0,1)} - \beta_{\bar{I},s}^{(0,1)}$$ and $$\beta_{\bar{I},{s-2}}^{(0,j+1)} -\beta_{\bar{I},{s-1}}^{(0,j+1)} \geq \beta_{\bar{I},{s-1}}^{(0,j+1)} -  \beta_{\bar{I},{s}}^{(0,j+1)}. $$  

	\item If $s\geq k_0-\ell_0+2$, there is nothing to show, since $\widetilde{\beta}^{(0,j)} = \beta^{(1,j+1)}$ for any $0\leq j \leq \mm -2$.
	\item In the other case, the result is ensured by condition \ref{part_c} on $\beta^{(0,j)}$.
\end{enumerate}

\item[{\textnormal{(b3)}}]\label{thruoghBlock} Let $1\leq j\leq \mm -1$. For $i=0$, we have that  
\begin{align*}
    &\widetilde{\beta}_{k_0-1}^{(0,j)} = \beta_{k_1-1}^{(1,j+1)}, \qquad  \widetilde{\beta}_{k_0}^{(0,j)} = \beta_{k_1}^{(1,j+1)},\\
    &\widetilde{\beta}_{1}^{(1,j)} =  \beta_{1}^{(2,j+1)}, \qquad \widetilde{\beta}_{2}^{(1,j)} =  \beta_{2}^{(2,j+1)}.
\end{align*} 
Hence, to verify that $$\widetilde{\beta}^{(0,j)}_{k_0-1} - \widetilde{\beta}^{(0,j)}_{k_0} \geq \widetilde{\beta}^{(0,j)}_{k_0} - \widetilde{\beta}^{(1,j)}_1 +1\geq \widetilde{\beta}^{(1,j)}_2 - \widetilde{\beta}^{(1,j)}_1 +1, $$ we need to check that 
$$\beta^{(1,j+1)}_{k_1-1} - \beta^{(1,j+1)}_{k_1} \geq \beta^{(1,j+1)}_{k_1} - \beta^{(2,j+1)}_1 +1 \geq \beta^{(2,j+1)}_2 - \beta^{(2,j+1)}_1 +1,$$ that is true by assumption.

\item[{\textnormal{(b4)}}] It follows from the expression \eqref{eq:newexponent}, by using the same reasoning of part \ref{part_b}.
\item[{\textnormal{(b5)}}] It follows by using the same reasoning of part \ref{part_c}.
\end{enumerate}

\end{enumerate}
\end{proof}


\begin{lemma}\label{lem:betaexponent}
For any $I\subseteq [\ell_0+r]$, such that $|I|=\ell_0$ and $I_0=[\ell_0]$, it holds that  
$\beta^{(0,j)}_{\bar{I}_0}\leq \beta^{(0,j)}_{\bar{I}}$ componentwise, for any $j$.
\begin{proof}
The proof immediately follows from the conditions  \ref{part_a}--\ref{part_e}.
\end{proof}
\end{lemma}


Let $L=(\ell_0,\dots,\ell_\mm )$, $K = (k_0,\dots,k_\mm )$. We are going to estimate the minimum and the maximum degrees in $x$ of the determinant of the matrix $\Gxy$. We define
\begin{align*}
\dd_{\mathrm{min}}(L,K,\mathrm{B},\Lambda) &:= \min \deg_x(\det(\Gxy))\in \N \cup \{\infty\},\\
\dd_{\mathrm{max}}(L,K,\mathrm{B},\Lambda) &:= \max \deg_x(\det(\Gxy))\in \N \cup \{\infty\}.
\end{align*} 
In the following lemma, we observe that, whenever $\Lambda$ is made of vectors satisfying the conditions \ref{lambdagrows}--\ref{lambda_ii}, $\dd_{\mathrm{min}}(L,K,\mathrm{B},\Lambda)$  and $\dd_{\mathrm{max}}(L,K,\mathrm{B},\Lambda)$  only depend on the fixed $L,K$ and $\mathrm{B}$.

\begin{lemma}\label{lem:mindegree}
Let $\Phi = (\phi^{(i,j)})_{{i,j}}$ and $\mathrm{B} = (\beta^{(i,j)})_{i,j}$, such that $\phi^{(i,j)}, \beta^{(i,j)}$ satisfy conditions \ref{part_a}--\ref{part_e} and $\phi^{(i,j)}\geq \beta^{(i,j)}$ componentwise for any $i,j$. Let $\Lambda = (\lambda^{(i,j)})_{i,j}$ be fixed, such that $\lambda^{(i,j)}$ satisfy conditions \ref{lambdagrows}--\ref{lambda_ii} for any $i,j$. Then, for any $\Upsilon = (\upsilon^{(i,j)})_{i,j}$, such that $\upsilon^{(i,j)}$ satisfies conditions \ref{lambdagrows}--\ref{lambda_ii}, we have
\begin{align*}
\dd_{\mathrm{min}}(L,K,\Phi,\Upsilon) & \geq \dd_{\mathrm{min}}(L,K,\mathrm{B},\Lambda), \\
\dd_{\mathrm{max}}(L,K,\Phi,\Upsilon) & \geq \dd_{\mathrm{max}}(L,K,\mathrm{B},\Lambda).
\end{align*}

\begin{proof}
Observe that the determinant of $\Gxy$ is a polynomial in $x$ with coefficients in $\F_q[\Y]$ and we can express it via the Leibniz formula. Let $N = \sum_{i=0}^\mm  k_i = \sum_{i=0}^\mm \ell_i$ and $\SN$ be the symmetric group of order $N$. Then
\[\det\Gxy = \sum_{\sigma\in\SN}\sgn(\sigma)\prod_{i=1}^N\Gxy_{i,\sigma(i)} = \sum_{\sigma\in\SN}\sgn(\sigma)G^\sigma,\]
where $G^{\sigma} = \prod_{i=1}^N\Gxy_{i,\sigma(i)} = R_{\sigma,\Lambda}(\Y)x^{s_{\sigma,\mathrm{B}}}$, for a suitable $s_{\sigma,\mathrm{B}}\in\N$.
Here, $R_{\sigma,\Lambda}(\Y)$ is a monomial in $\Y$, which can also be $0$, depending on $\sigma$.  Hence,
$$\det\Gxy = \sum_{\sigma\in\SN}R_{\sigma,\Lambda}(\Y)x^{s_{\sigma,\mathrm{B}}} =\sum_{\sigma\in Z^\Lambda}R_{\sigma,\Lambda}(\Y)x^{s_{\sigma,\mathrm{B}}},$$
where $Z^\Lambda:=\{\sigma\in \SN \mid G^{\sigma} \ne 0\}.$ Obviously, as $\sigma(i)=i\in Z^\Lambda,\ Z^\Lambda\neq \emptyset$.

Observe that the elements $R_{\sigma,\Lambda}(\Y)$, with  $\sigma\in Z^\Lambda$ are $\F_q$-linearly independent. Indeed, they are monomials in $\Y$ and they all involve distinct exponents due to conditions \ref{lambdagrows}--\ref{lambda_ii} on the components of $\Lambda$.  This remark is crucial for the rest of the proof.

\noindent
Let 
\begin{gather*}
P_{\Phi, \Upsilon}(x,\Y) =\det(G(x,\Y,\Phi, \Upsilon)) = \sum_{\sigma\in Z^{\Upsilon}} R_{\sigma,\Upsilon}(\Y)x^{s_{\sigma,\Phi}},\\
P_{\mathrm{B}, \Lambda}(x,\Y) =\det(G(x,\Y,\mathrm{B},\Lambda)) = \sum_{\sigma\in Z^{\Lambda}} R_{\sigma,\Lambda}(\Y)x^{s_{\sigma,\mathrm{B}}}.
\end{gather*}
The only thing to observe is that $R_{\sigma,\Upsilon}(\Y) = 0$
if and only if $R_{\sigma,\Lambda}(\Y) = 0$.
This is true because of the previous observation. Indeed, these monomials in $\Y$ are linearly independent over $\F_q$, hence their nonzeroness depends only on the support of the matrix and not on their exponents. This implies that $Z^\Upsilon = Z^\Lambda$. In particular, if $R_{\sigma,\Lambda}(\Y)x^{\sigma,\mathrm{B}}$ is the monomial in $\F_q[\Y][x]$ corresponding to the minimum degree in $x$ of $P_{\mathrm{B},\Lambda}(x,\Y)$, then the monomial 
$R_{\sigma,\Upsilon}(\Y)x^{s_{\sigma,\Phi}}$ corresponds to the minimum degree in $x$ of  $P_{\Phi, \Upsilon}(x,\Y)$.
 Furthermore, by our assumptions on $\Phi$ and $\mathrm{B}$, we have  $s_{\sigma,\mathrm{B}}\leq s_{\sigma,\Phi}$. Hence, $\dd_{\mathrm{min}}(L,K,\Phi,\Upsilon) \geq \dd_{\mathrm{min}}(L,K,\mathrm{B},\Lambda).$ 
The same argument also holds for the maximum degree.
\end{proof}
\end{lemma}

We are now ready to present how to determine the monomial of minimum  degree in $x$ of 
$\det(\Gxy)$.

\begin{theorem}\label{thm:main}
The determinant of $\Gxy$ is nonzero. Moreover, the monomial with minimum degree  in $x$ of $\det(\Gxy)$\footnote{For monomial with minimum degree in $x$ of $\det(\Gxy)$, we mean the monomial as an element in $\F_q[\Y][x]$, of the form $p(\Y)x^b$, where $p(\Y)$ is in the ring of coefficients $\F_q[\Y]$ and $b$ is the smallest exponent of $x$ involved in the expression of $\det(\Gxy)$.}   is given by the product of the 
$\ell_i\times\ell_i$ minors across the main diagonal, for $i=0,\ldots, \mm $.   More precisely, let $L_0$ be the set of the smallest $\ell_0$ row indices of $A_{0,0}$, \ie  $L_0= \{1,2,\dots,\ell_0\}$. For any $i=1,\dots,\mm $ define $L_i$ to be the set of the smallest $\ell_i$ row indices corresponding to the $i$-th column block of $\Gxy$ after deleting the rows indexed by $\cup_{j=0}^{i-1} L_j$, \ie $L_i$ is given by the first $\ell_i$ indices in $\{1,\dots, \sum_{j=0}^i k_j\}\setminus \cup_{j=0}^{i-1} L_j$. 
Then, the monomial with minimum degree in $x$ of $\det(\Gxy)$  is the product, for $i=0,\ldots, \mm $, of the $\ell_i \times \ell_i$ minors whose rows are indexed by $L_i$ and whose columns are the one corresponding to the $i$-th column block for $i=0,\dots,\mm $.
\end{theorem}

\begin{proof}
Recall that $$\Gxy = \begin{pmatrix}
A_{0,0} & A_{0,1} & \cdots & A_{0,\mm } \\
 & A_{1,1} &\cdots & A_{1,\mm } \\
 & & \ddots & \vdots \\
& & & A_{\mm ,\mm }
\end{pmatrix},$$ 
with $A_{i,j} = A_{i,j}^{\left(\beta^{(i,j)},\lambda^{(i,j)}\right)}$ and that $\lambda^{(i,j)}$ and $\beta^{(i,j)}$ satisfy the conditions \ref{lambdagrows}--\ref{lambda_ii} and \ref{part_a}--\ref{part_e} given in Definition \ref{def:shapeG}.

We will prove the result by induction on $\mm $.

\noindent
\underline{\textbf{Base case}} $\mm =0$: In this case, $k_0 = \ell_0$ and $A_{0,0}$ is a generalized Vandermonde matrix, whose determinant is a nonzero polynomial.

\noindent
\underline{\textbf{Induction case:}} Assume that the result is true for all the numbers of blocks up to $\mm $ and prove it for $\mm +1$ blocks.

This time we are going to use Laplace formula to compute the determinant of $\Gxy$.

\begin{align*}\label{eq:determinant}
\det(\Gxy) &= \sum_{\substack{I\subseteq [\ell_0+r]\\ |I| = \ell_0}}(\pm 1) \det(A_I(Y))\det(\GI)  \\
&= \sum_{\substack{I\subseteq [\ell_0+r]\\ |I| = \ell_0}}(\pm 1) \det(A_I(Y))\det\begin{pmatrix}
\diag\left(x^{\beta_{\bar{I}}^{(0,1)}}\right) & 0 \\
0 & \mathrm{Id}
\end{pmatrix}\det(\widetilde{G}_{\bar{I}}(x,\Y, \mathrm{B}_{\bar{I}},\Lambda_{\bar{I}}))   \\
&= \sum_{\substack{I\subseteq [\ell_0+r]\\ |I| = \ell_0}}(\pm 1) \det(A_I(Y))x^{b_{\bar{I}}}\det(\widetilde{G}_{\bar{I}}(x,\Y, \mathrm{B}_{\bar{I}},\Lambda_{\bar{I}}))\\
&=\det(A_{I_0}(Y))x^{b_{\bar{I}_0}}\det(\widetilde{G}_{\bar{I}_0}(x,\Y, \mathrm{B}_{\bar{I}_0},\Lambda_{\bar{I}_0})) + \\
&  \sum_{\substack{I\subseteq [\ell_0+r]\\ |I| = \ell_0 \\ I\ne [\ell_0]}}(\pm 1)\det(A_I(Y))x^{b_{\bar{I}}}\det(\widetilde{G}_{\bar{I}}(x,\Y, \mathrm{B}_{\bar{I}},\Lambda_{\bar{I}})).
\end{align*}

Observe that $\widetilde{G}_{\bar{I}}(x,\Y, \mathrm{B}_{\bar{I}},\Lambda_{\bar{I}})$ is composed by $\mm \times \mm$ blocks and by Lemma \ref{lemma:shapeGtilde}, 
the vectors in $\mathrm{B}_{\bar{I}}$ and $\Lambda_{\bar{I}}$ satisfy  conditions \ref{lambdagrows}--\ref{lambda_ii} and \ref{part_a}--\ref{part_e}, that is $\widetilde{G}_{\bar{I}}(x,\Y, \mathrm{B}_{\bar{I}},\Lambda_{\bar{I}})$ is of the form given in 
Definition \ref{def:shapeG}.

Let $M := \dd_{\mathrm{min}}(L_{\bar{I}_0},K_{\bar{I}_0},\mathrm{B}_{\bar{I}_0},\Lambda_{\bar{I}_0})$.
Now, from Lemma \ref{lem:betaexponent} and Lemma \ref{lem:mindegree} we have that $M\leq \dd_{\mathrm{min}}(L_{\bar{I}},K_{\bar{I}},\widetilde{\mathrm{B}}_{\bar{I}},\widetilde{\Lambda}_{\bar{I}})$, for any $I\subseteq[\ell_0+r]$, with $|I|=\ell_0$. Hence, we can write $\dd_{\mathrm{min}}(L_{\bar{I}},K_{\bar{I}},\mathrm{B}_{\bar{I}},\Lambda_{\bar{I}}) = M+s_I$, where $s_I\in\N$ depends on the chosen $I$.

Therefore, 
\begin{align*}
\det(\Gxy)&= x^{b_{\bar{I}_0}}\det(A_{I_0}(Y))x^M\left(P_{I_0}(Y) + xQ_{I_0}(x,Y)\right) + \\ &+\sum_{\substack{I\subseteq [\ell_0+r]\\ |I| = \ell_0 \\ I\ne [\ell_0]}}(\pm 1)x^{b_{I}}x^{M+s_{I}}\det(A_I(Y))\left(P_{I}(Y) + xQ_{I}(x,Y)\right),
\end{align*}
where $P_{I_0}(Y), P_{I}(Y)\ne 0$, for any $I$.

\noindent
By definition of $b_{\bar{I}}$ and $\beta^{(i,j)}$, it is evident that $b_{\bar{I}_0} < b_{\bar{I}}$ for any $I\ne I_0$. Hence, we have that $M$ is the minimum degree of the determinant of $\GIbar$ and, by inductive hypothesis, the corresponding monomial is obtained by multiplying the $\ell_i\times \ell_i$ minors across the main diagonal for $i=1,\dots, \mm $.
Therefore, the minimum degree in $x$ of $\det(\Gxy)$ is ${b_{\bar{I}_{0}}+M}$   and its corresponding monomial is given by the product of the determinants obtained by selecting the first $\ell_0$ rows and columns and the $\ell_i\times\ell_i$ minors across the main diagonal, as explained in the statement. 
\end{proof}

We now provide an exhaustive example from  which will illustrate Theorem \ref{thm:main}.

\begin{example}\label{exa:Gxy_minimal_degree}
Next, we show an example of how to construct the matrix $G(x,\Y, \mathrm{B}, \Lambda)$ given in Definition~\ref{def:shapeG} 
and use Theorem~\ref{thm:main} to determine the monomial of minimum degree in $x$ of $\det(\Gxy)$.

Let  $\mm=2$ and $(k_0,k_1,k_2)=(3,4,5), (\ell_0,\ell_1,\ell_2)=(2,4,6) \in \N^3$. The blocks $A_{i,j}$ with $0 \leq i \leq j \leq 2$ composing $\Gxy$ are of the form
\[
\begin{array}{ll}
A_{i,j} & =  \mathrm{diag}\left(x^{\beta^{(i,j)}}\right)V\left(\lambda^{(i,j)},\y^{(j)}\right)\\
  & = 
\begin{pmatrix}
x^{\beta_1^{(i,j)}}\left(y_1^{(j)}\right)^{\lambda_1^{(i,j)}} & x^{\beta_1^{(i,j)}}\left(y_2^{(j)}\right)^{\lambda_1^{(i,j)}} & \cdots & x^{\beta_1^{(i,j)}}\left(y_{\ell_j}^{(j)}\right)^{\lambda_1^{(i,j)}} \\
x^{\beta_2^{(i,j)}}\left(y_1^{(j)}\right)^{\lambda_2^{(i,j)}} & x^{\beta_2^{(i,j)}}\left(y_2^{(j)}\right)^{\lambda_2^{(i,j)}} & \cdots & x^{\beta_2^{(i,j)}}\left(y_{\ell_j}^{(j)}\right)^{\lambda_2^{(i,j)}} \\
\vdots & \vdots & \ddots & \vdots\\
x^{\beta_{k_i}^{(i,j)}}\left(y_1^{(j)}\right)^{\lambda_{k_i}^{(i,j)}}& x^{\beta_{k_i}^{(i,j)}}\left(y_2^{(j)}\right)^{\lambda_{k_i}^{(i,j)}} & \cdots & x^{\beta_{k_i}^{(i,j)}}\left(y_{\ell_j}^{(j)}\right)^{\lambda_{k_i}^{(i,j)}} \\
\end{pmatrix}\\
\end{array}.
\]
Consider $\Y:=(y^{(0)},y^{(1)},y^{(2)})$ the vector of variables, where
 \[
\begin{array}{lll}
y^{(0)}& = & (y_1^{(0)},y_2^{(0)}) = (y_1,y_2),\\
y^{(1)}& = & (y_1^{(1)},y_2^{(1)},y_3^{(1)},y_4^{(1)})= (z_1,z_2,z_3,z_4),\\
y^{(2)}& = & (y_1^{(2)},y_2^{(2)},y_3^{(2)},y_4^{(2)},y_5^{(2)},y_6^{(2)})= (w_1,w_2,w_3,w_4,w_5,w_6).\\
\end{array}
\]
Let  $\beta^{(i,j)}=\left(\beta^{(i,j)}_1, \ldots, \beta^{(i,j)}_{k_i} \right)$ be the vector consisting of the powers of $x$ of the rows of $A_{i,j}$, with $0 \leq i \leq j \leq 2$. We only take the values for a fixed column of $A_{i,j}$, since that every column has the same powers. Let
$$\mathrm{B}:=(\beta^{(0,0)},\beta^{(0,1)},\beta^{(0,2)},\beta^{(1,1)},\beta^{(1,2)},\beta^{(2,2)}),$$ where, in order to lighten the notation, we define
$$
\begin{array}{rclrclrcl}
\beta^{(0,0)}& = & 
(0,0,0) &
\beta^{(0,1)}& = & 
(2,1,0) & 
\beta^{(0,2)}& = & 
(9,7,5)\\
& & & \beta^{(1,1)}& = & 
(0,0,0,0) &
\beta^{(1,2)}& = & 
(3,2,1,0) \\
& & & & & & \beta^{(2,2)}& = & 
(0,0,0,0,0).
\end{array}
$$
Now, let $\lambda^{(i,j)}=\left(\lambda^{(i,j)}_1, \ldots, \lambda^{(i,j)}_{k_i} \right)$ be the vector composed by the powers of variables $Y$ of a fixed column of $A_{i,j}$, with $0 \leq i \leq j \leq 2$. In our case, we obtain that  $$\Lambda:=\left(\lambda^{(0,0)},\lambda^{(0,1)},\lambda^{(0,2)},\lambda^{(1,1)},\lambda^{(1,2)},\lambda^{(2,2)} \right),$$ where
\[
\begin{array}{rclrclrcl}
\lambda^{(0,0)}& = & 
(2,1,0) &
\lambda^{(0,1)}& = & 
(6,5,4) &
\lambda^{(0,2)}& = & 
(11,10,9)\\
& & & \lambda^{(1,1)}& = & 
(3,2,1,0)&
\lambda^{(1,2)}& = & 
(8,7,6,5)\\
& & & & & & \lambda^{(2,2)}& = &
(4,3,2,1,0).
\end{array}
\]
It is easy to check that the values $\beta^{(i,j)}$ and $\lambda^{(i,j)}$ of the vectors $\mathrm{B}$ and $\Lambda$, respectively, satisfy the conditions of Definition~\ref{def:shapeG}.

Hence, the matrix $\Gxy$ is given by

\begin{gather*}
G(x,\Y, \mathrm{B}, \Lambda)
 =  
 \begin{pmatrix}
A_{0,0} & A_{0,1} &  A_{0,2} \\
        & A_{1,1} & A_{1,2} \\
        &         & A_{2,2}
\end{pmatrix}= \\
\left(
\begin{array}{cc:cccc:cccccc}
y_1^2 & y_2^2 &  x^2z_1^6 & x^2z_2^6 & x^2z_3^6 & x^2z_4^6 & x^9w_1^{11} & x^9w_2^{11} & x^9w_3^{11} & x^9w_4^{11} & x^9w_5^{11} & x^9w_6^{11}\\
y_1 & y_2 & xz_1^5 & xz_2^5 & xz_3^5 & xz_4^5 & x^7w_1^{10} & x^7w_2^{10} & x^7w_3^{10} & x^7w_4^{10} & x^7w_5^{10} & x^7w_6^{10}\\
1   & 1   & z_1^4 & z_2^4 & z_3^4 & z_4^4 & x^5w_1^{9} & x^5w_2^{9} & x^5w_3^{9} & x^5w_4^{9} & x^5w_5^{9} & x^5w_6^{9}\\ \hdashline
    &      & z_1^3 & z_2^3 & z_3^3 & z_4^3 & x^3w_1^{8} & x^3w_2^{8} & x^3w_3^{8} & x^3w_4^{8} & x^3w_5^{8} & x^3w_6^{8}\\
        &      & z_1^2 & z_2^2 & z_3^2 & z_4^2 & x^2w_1^{7} & x^2w_2^{7} & x^2w_3^{7} & x^2w_4^{7} & x^2w_5^{7} & x^2w_6^{7}\\
           &      & z_1 & z_2 & z_3 & z_4 & xw_1^{6} & xw_2^{6} & xw_3^{6} & xw_4^{6} & xw_5^{6} & xw_6^{6}\\
           &      & 1   & 1   & 1   & 1   & w_1^{5} & w_2^{5} & w_3^{5} & w_4^{5} & w_5^{5} & w_6^{5}\\ \cdashline{3-12}
\multicolumn{6}{c:}{}                     &   w_1^{4} & w_2^{4} & w_3^{4} & w_4^{4} & w_5^{4} & w_6^{4}\\
\multicolumn{6}{c:}{}                     &  w_1^{3} & w_2^{3} & w_3^{3} & w_4^{3} & w_5^{3} & w_6^{3}\\
\multicolumn{6}{c:}{}                     &  w_1^{2} & w_2^{2} & w_3^{2} & w_4^{2} & w_5^{2} & w_6^{2}\\
\multicolumn{6}{c:}{}                     & w_1     & w_2     & w_3     & w_4     & w_5      & w_6\\           
\multicolumn{6}{c:}{}                     &  1       &    1    &   1     &  1      &  1     &  1           
\end{array}
\right).\\
\end{gather*}

In order to compute the minimum degree in $x$ of the determinant of $\Gxy$ we use Theorem~\ref{thm:main}. 
To this end, we need to define the sets $L_i$ for $i=0,1,2$, given in Theorem~\ref{thm:main}. $L_0$ is the set of the smallest $\ell_0$ row indices of $A_{0,0}$, that is $L_0=\{1,2\}$, 
and the rest of sets $L_i$ are given by the first $l_i$ indices in $\{1,\dots, \sum_{j=0}^i k_j\}\setminus \cup_{j=0}^{i-1} L_j$. First, we obtain these previous subsets for any $i=1,2$. $L_1$ is composed by the first $\ell_1=4$ indices in 
\[
\{1,\dots, k_0+k_1\}\setminus  L_0 =\{1,2,3,4,5,6,7\} \setminus \{1,2\} = \{3,4,5,6,7\},
\]
that is, $L_1=\{3,4,5,6\}$; and  $L_2$ is composed by the first $l_2$ indices in 
\[
\{1,\dots, k_0+k_1+k_2\}\setminus  \left(L_0 \cup L_1 \right) =\{1,2,\ldots,12\} \setminus \{1,2,3,4,5,6\} = \{7,8,9,10,11,12\},
\]
that is, $L_2=\{7,8,9,10,11,12\}$.
By applying Theorem ~\ref{thm:main}, we have that the monomial of minimal degree in $x$ of $\det(\Gxy)$ is the product of three square minors, whose rows are indexed by $L_i$ for $i=0,1,2$, 
given in the diagonal of the matrix $\Gxy$

\[
\left(\begin{array}{|cc|cccc|cccccc|}\hhline{|--|}
\cellyellow {y_1^2} & \cellyellow{y_2^2} & x^2z_1^6 & x^2z_2^6 & x^2z_3^6 & \multicolumn{1}{c}{x^2z_4^6} & x^9w_1^{11} & x^9w_2^{11} & x^9w_3^{11} & x^9w_4^{11} & x^9w_5^{11} & \multicolumn{1}{c}{x^9w_6^{11}}\\
\cellyellow {y_1} & \cellyellow {y_2} & xz_1^5  & xz_2^5 & xz_3^5 & \multicolumn{1}{c}{xz_4^5} & x^7w_1^{10} & x^7w_2^{10} & x^7w_3^{10} & x^7w_4^{10} & x^7w_5^{10} & \multicolumn{1}{c}{x^7w_6^{10}}\\  \hhline{--|----|}
\multicolumn{1}{c}{1}   & 1 & \cellblue {z_1^4} & \cellblue {z_2^4} & \cellblue {z_3^4} & \cellblue {z_4^4} & x^5w_1^{9} & x^5w_2^{9} & x^5w_3^{9} & x^5w_4^{9} & x^5w_5^{9} & \multicolumn{1}{c}{x^5w_6^{9}} \\ 
 \multicolumn{1}{c}{}   &      & \cellblue {z_1^3} & \cellblue {z_2^3} & \cellblue {z_3^3} & \cellblue {z_4^3} & x^3w_1^{8} & x^3w_2^{8} & x^3w_3^{8} & x^3w_4^{8} & x^3w_5^{8} & \multicolumn{1}{c}{x^3w_6^{8}}\\
    \multicolumn{1}{c}{}    &      & \cellblue {z_1^2} & \cellblue {z_2^2} & \cellblue {z_3^2} & \cellblue {z_4^2} & x^2w_1^{7} & x^2w_2^{7} & x^2w_3^{7} & x^2w_4^{7} & x^2w_5^{7} & \multicolumn{1}{c}{x^2w_6^{7}}\\
   \multicolumn{1}{c}{}       &      & \cellblue{z_1} & \cellblue {z_2} & \cellblue {z_3} & \cellblue {z_4} & xw_1^{6} & xw_2^{6} & xw_3^{6} & xw_4^{6} & xw_5^{6} & \multicolumn{1}{c}{xw_6^{6}}\\ \hhline{~~|----|------|}
\multicolumn{2}{c}{} &  1   & 1   & 1   & 1  & \cellred {w_1^{5}} & \cellred {w_2^{5}} & \cellred {w_3^{5}} & \cellred {w_4^{5}} & \cellred {w_5^{5}} & \cellred {w_6^{5}}\\
\multicolumn{6}{c|}{}                 &  \cellred {w_1^{4}} & \cellred {w_2^{4}} & \cellred {w_3^{4}} & \cellred {w_4^{4}} & \cellred {w_5^{4}} & \cellred {w_6^{4}}\\
\multicolumn{6}{c|}{}          & \cellred { w_1^{3}} & \cellred {w_2^{3}} & \cellred {w_3^{3}} & \cellred {w_4^{3}} & \cellred {w_5^{3}} & \cellred {w_6^{3}}\\
\multicolumn{6}{c|}{}                  &  \cellred {w_1^{2}} & \cellred {w_2^{2}} & \cellred {w_3^{2}} & \cellred {w_4^{2}} & \cellred {w_5^{2}} & \cellred {w_6^{2}}\\
\multicolumn{6}{c|}{}                &  \cellred {w_1}     & \cellred {w_2}     & \cellred {w_3}     & \cellred {w_4}     & \cellred {w_5}      & \cellred {w_6}\\           
\multicolumn{6}{c|}{}                 &  \cellred {1}       &    \cellred {1}    &   \cellred {1}     &  \cellred {1}      &  \cellred {1}     &  \cellred {1}  \\ \cline{7-12}
\end{array}\right)
\]

This product of determinants produces the monomial in 
$\F_q[\Y][x]$ of minimal degree in $x$ (that is $0$) given by
$$ y_1y_2z_1z_2z_3z_4(y_1-y_2)\prod_{1\leq i <j \leq 4}(z_i-z_j) \prod_{1\leq i < j \leq 6}(w_i-w_j).$$
\end{example}

We also obtain a similar  result for the maximum degree. This time the monomial of maximum degree is obtained by taking the product of another set of minors. However, we will omit the proof, since it is technical and it is based on the same idea of the proof of Theorem \ref{thm:main}. Using induction and Laplace formula for computing the determinant, one reduces to the case of one block less, estimating the degrees of monomials by means of Lemma \ref{lem:mindegree}.

\begin{theorem}\label{thm:maxdegree} Let $L_0$ be the set of the highest $\ell_0$ row indices of $A_{0,0}$, \ie \  $L_0= \{k_0-\ell_0+1,k_0-\ell_0+2,\dots,k_0\}$. For any $i=1,\dots,\mm $ define $L_i$ to be the set of the highest $\ell_i$ row indices corresponding to the $i$-th column block of $\Gxy$ after deleting the rows indexed by $\cup_{j=0}^{i-1} L_j$, \ie $L_i$ is given by the highest $\ell_i$ indices in $\{1,\dots, \sum_{j=0}^i k_j\}\setminus \cup_{j=0}^{i-1} L_j$. 
Then, the monomial with maximum  degree in $x$ of $\det(\Gxy)$  is the product, for $i=0,\ldots, \mm $, of the $\ell_i \times \ell_i$ minors whose rows are indexed by $L_i$ and whose columns are the one corresponding to the $i$-th column block for $i=0,\dots,\mm$.
\end{theorem}

\begin{example}
 Following the same setting as the one in Example \ref{exa:Gxy_minimal_degree}, we illustrate how Theorem \ref{thm:maxdegree} works. In this case, we have that the monomial with maximal degree in $x$ of $\Gxy$ is obtained taking $L_0=\{2,3\}$, $L_1=\{4,5,6,7\}$ and $L_2=\{1,8,9,10,11,12\}$. Graphically, the product of the three minors are given below
  \[
\left(
\begin{array}{|cc|cccc|cccccc|} \hhline{~~~~~~|------|}
\multicolumn{1}{c}{y_1^2} & \multicolumn{1}{c}{y_2^2} & x^2z_1^6 & x^2z_2^6 & x^2z_3^6 & x^2z_4^6 & \cellred {x^9w_1^{11}} & \cellred {x^9w_2^{11}} & \cellred {x^9w_3^{11}} & \cellred {x^9w_4^{11}} & \cellred {x^9w_5^{11}} & \cellred {x^9w_6^{11}}\\  \hhline{|--|~~~~|------|}
\cellyellow {y_1} & \cellyellow {y_2} & xz_1^5  & xz_2^5 & xz_3^5 & \multicolumn{1}{c}{xz_4^5} & x^7w_1^{10} & x^7w_2^{10} & x^7w_3^{10} & x^7w_4^{10} & x^7w_5^{10} & \multicolumn{1}{c}{x^7w_6^{10}}\\  
\cellyellow {1}   & \cellyellow {1} & {z_1^4} & {z_2^4} & {z_3^4} & \multicolumn{1}{c}{z_4^4} & x^5w_1^{9} & x^5w_2^{9} & x^5w_3^{9} & x^5w_4^{9} & x^5w_5^{9} & \multicolumn{1}{c}{x^5w_6^{9}} \\ \hhline{|--|----|}
 \multicolumn{1}{c}{}   &      & \cellblue {z_1^3} & \cellblue {z_2^3} & \cellblue {z_3^3} & \cellblue {z_4^3} & x^3w_1^{8} & x^3w_2^{8} & x^3w_3^{8} & x^3w_4^{8} & x^3w_5^{8} & \multicolumn{1}{c}{x^3w_6^{8}}\\
  \multicolumn{1}{c}{}      &      & \cellblue {z_1^2} & \cellblue {z_2^2} & \cellblue {z_3^2} & \cellblue {z_4^2} & x^2w_1^{7} & x^2w_2^{7} & x^2w_3^{7} & x^2w_4^{7} & x^2w_5^{7} & \multicolumn{1}{c}{x^2w_6^{7}}\\
   \multicolumn{1}{c}{}        &      & \cellblue {z_1} & \cellblue {z_2} & \cellblue {z_3} & \cellblue {z_4} & xw_1^{6} & xw_2^{6} & xw_3^{6} & xw_4^{6} & xw_5^{6} & \multicolumn{1}{c}{xw_6^{6}}\\ 
\multicolumn{1}{c}{} & & \cellblue {1}   & \cellblue {1}   & \cellblue {1}   & \cellblue {1}  & {w_1^{5}} & {w_2^{5}} & {w_3^{5}} & {w_4^{5}} & {w_5^{5}} & \multicolumn{1}{c}{w_6^{5}}\\ \hhline{|~~|----|------|}
\multicolumn{6}{c|}{}                     &  \cellred {w_1^{4}} & \cellred {w_2^{4}} & \cellred {w_3^{4}} & \cellred {w_4^{4}} & \cellred {w_5^{4}} & \cellred {w_6^{4}}\\
\multicolumn{6}{c|}{}                     & \cellred { w_1^{3}} & \cellred {w_2^{3}} & \cellred {w_3^{3}} & \cellred {w_4^{3}} & \cellred {w_5^{3}} & \cellred {w_6^{3}}\\
\multicolumn{6}{c|}{}                     &  \cellred {w_1^{2}} & \cellred {w_2^{2}} & \cellred {w_3^{2}} & \cellred {w_4^{2}} & \cellred {w_5^{2}} & \cellred {w_6^{2}}\\
\multicolumn{6}{c|}{}                     &  \cellred {w_1}     & \cellred {w_2}     & \cellred {w_3}     & \cellred {w_4}     & \cellred {w_5}      & \cellred {w_6}\\           
\multicolumn{6}{c|}{}                     &  \cellred {1}       &    \cellred {1}    &   \cellred {1}     &  \cellred {1}      &  \cellred {1}     &  \cellred {1}  \\ \cline{7-12}
\end{array}
\right).
\]
\end{example}

We are now ready to give a proof of Theorem \ref{thm:ConstructionMDP}, which  is based on the result of Theorem \ref{thm:main}.

\begin{proof}[Proof of Theorem \ref{thm:ConstructionMDP}] 
Let $\C$ be the $(n,k,\delta)_{q^s}$ convolutional code $\C_{k,n}^\delta(\gamma, \alpha)$. Set $L:=\floor{\frac{\delta}{k}}+\floor{\frac{\delta}{n-k}}$ and let 
$$G_L^\cc :=\begin{pmatrix}
G_0 & G_1 & \cdots  & G_L \\
 & G_0 & \cdots & G_{L-1} \\
 & & \ddots & \vdots \\
 & &  & G_0\\
\end{pmatrix}$$
be the  $L$-th truncated sliding generator matrix.
Moreover, for $i=0,\ldots, L$, let $G_i(x)$ be the matrix $G_i$ where we have substituted the element $\gamma \in \F_{q^s}$ with an algebraically independent variable $x$,
and $G_L^\cc(x)$ be corresponding $L$-th truncated sliding generator matrix. With this notation, we  have that $G_i(\gamma)=G_i$ for $i=0,\ldots, L$ and hence $G_L^\cc(\gamma)=G_L^\cc$. 
Moreover, by Lemma \ref{lem:MDPpolynomial}, we only need to prove that $p(\gamma) \neq 0$ for every $p(x) \in \mathcal P(k,n,\delta,\alpha)$.

\noindent Now we divide the proof in  two distinct cases.

\underline{\textbf{Case I: $\delta=km$}}. In this case, every matrix $G_i$ appearing as a block of $G_L^\cc$ is either $M_i$ or the zero matrix. With this setting, we analyze two subcases.

\underline{\textbf{Case I-A: $km<n-k$}}. We consider the column blocks of $G_L^\cc(x)$ indexed by $0,1,\ldots, L$. Let $F(x)$ be a square $(L+1)k\times (L+1)k$ submatrix of $G_L^\cc(x)$ obtained by selecting $\bar{\ell}_i$ columns from the $i$-th column block, where the $\bar\ell_i$'s satisfy \eqref{eq:lis}. Let $I_F:=\{t_0,t_1,\ldots,t_\mm\}$ be the set of indices of the column blocks involved in the selection of $F(x)$ (\ie $a\in I_F$ if and only if there exists a column of $F(x)$ which comes from the $a$-th column block), where we have ordered them as $0 \leq t_0<\ldots < t_\mm \leq \mm$. At this point, one can see that $F(x)$ is a block upper triangular matrix of the form
$$ F(x)=\begin{pmatrix}
A_{0,0} & A_{0,1} & \cdots & A_{0,\mm } \\
 & A_{1,1} &\cdots & A_{1,\mm } \\
 & & \ddots & \vdots \\
& & & A_{\mm ,\mm }
\end{pmatrix},$$
where $A_{0,j}\in \Fq[x]^{((t_0+1)k)\times \bar\ell_{t_j}}$, and $A_{i,j}\in \Fq[x]^{((t_i-t_{i-1})k)\times \bar\ell_{t_j}}$ for $1\leq i \leq \mm$ and $0\leq j \leq \mm$. Now, we are going to show that $F(x)$ is obtained from a matrix of the form $\Gxy$, after a suitable specialization of $Y$. Let us define $k_0:=(t_0+1)k$, $k_i:=(t_i-t_{i-1})k$  for $1\leq i \leq \mm$ and $\ell_j:=\bar\ell_{t_j}$ for $0\leq j \leq \mm$. One can see that for each $s$ such that $0\leq s \leq \mm-1$, we have
$$\sum_{j=0}^s\ell_j=\sum_{j=0}^{t_s}\bar\ell_j\leq (t_s+1)k=(t_0+1)k+\sum_{j=1}^s(t_j-t_{j-1})k=\sum_{j=0}^sk_j,$$
and that
$$ \sum_{j=0}^\mm\ell_j=\sum_{j=0}^{t_\mm}\bar\ell_j=(L+1)k=\sum_{j=0}^\mm k_j.$$ 
By definition, in each block $A_{i,j}$, every row is a monomial in $x$ with constant degree, and hence it can be written as $\diag(x^{\beta^{(i,j)}})V(\lambda^{(i,j)},\alpha^{(j)})$, for suitable  $\beta^{(i,j)}, \lambda^{(i,j)}\in\N^{k_i}$, and some $\alpha^{(j)}\in\Fq^{\ell_j}$ obtained   selecting $\ell_j$ entries from $\alpha$. 
Lengthy computations show that the vectors $\beta^{(i,j)}$'s satisfy the conditions \ref{lambdagrows}--\ref{lambda_ii} and that the vectors $\lambda^{(i,j)}$'s satisfy the conditions \ref{part_a}--\ref{part_e}. Define now $\widehat{\mathrm{B}}:=({\beta}^{(i,j)})_{i,j}, \widehat{\Lambda}:=({\lambda}^{(i,j)})_{i,j}$, $\mathrm{A}=(\alpha^{(j)})_j$. Hence, 
our submatrix $F(x)$ is obtained from the matrix $G(x,\Y, \widehat{\mathrm{B}},\widehat{\Lambda})$ by evaluating the $\Y$ in a suitable vector $\mathrm{A}$ of $\alpha_i$'s, \ie $F(x)=G(x,\mathrm{A}, \widehat{\mathrm{B}},\widehat{\Lambda})$.

By Theorem \ref{thm:main}, the determinant of $G(x,\Y, \widehat{\mathrm{B}},\widehat{\Lambda})$ is a nonzero polynomial in $x$ and $\Y$, which we denote by $f(x,\Y)=\sum_{i}f_i(\Y)x^i$.  Therefore, $\det(F(x))=f(x,\mathrm{A})$. It remains to prove that $f(x,\mathrm{A})$ is still a nonzero polynomial, \ie that there exists at least one index $i$ such that $f_i(\mathrm{A})\neq0$.
By Theorem \ref{thm:main} we know that the monomial of minimum degree in $x$ is $f_b(\Y)x^b$, for some $b \in \mathbb N$, and it is obtained by multiplying the $\ell_i \times \ell_i$ minors along the main diagonal.   By the structure of these minors,
we have that for every $i\in \{0,\ldots, \mm\}$, there exist  integers $b_i,t_i\in \mathbb N$ and a set $J_i \subseteq [n]$ with $|J_i|=\ell_i$, such that the corresponding $\ell_i \times \ell_i$ minor is given by 
$$
\left(\prod_{j \in J_i} \alpha_j\right)^{t_i}x^{b_i}\det(V(\alpha^{(i)}, (0,1,\ldots, k_i-1))),
$$
where $J_i$ is  selected from $\alpha$ in order to get $\alpha^{(i)}$, and where $b=b_0 + \cdots + b_\mm$.
Since the matrices $V(\alpha^{(i)}, (0,1,\ldots, k_i-1))$ are classical Vandermonde matrices with pairwise distinct and nonzero defining entries, we get $f_b(\mathrm{A})\neq 0$ and so $f(x,\mathrm{A}) \neq 0$.

Notice that the product of the $\ell_i\times \ell_i$ minors along the main diagonal for $i=0,\ldots, \mm$ corresponds to the product of the $\bar\ell_i\times \bar\ell_i$ minors along the main diagonal for $i=0,\ldots, L$, if we use the convention that the determinant of a $0\times 0$  matrix is $1$.

This shows that the set $\mathcal P(k,n,km,\alpha)$ does not contain the zero-polynomial. Let $p(x)\in \mathcal P(k,n,km,\alpha)$. Since $p(x) \neq 0$, we can write $p(x)=x^{\nu (p(x))}p_1(x)$, where $p_1(x) \in \Fq[x]$ and $\deg (p_1(x))=\deg (p(x))-\nu (p(x))$. Since $s>\deg (p_1(x))$, we get that $p_1(\gamma) \neq 0$ and therefore $p(\gamma)\neq 0$.

\underline{\textbf{Case I-B: $km\geq n-k$}}.
Let $r:=\left\lfloor\frac{\delta}{n-k} \right\rfloor=\left\lfloor\frac{km}{n-k} \right\rfloor$. Then we have  $L=m+r$ and $G_{m+i}(x)=0$ for $i=1,\ldots,r$.
Therefore, the polynomial version $G_L^\cc(x)=G_{m+r}^\cc(x)$  of the $L$-th truncated sliding generator matrix is given by 

$$G_L^\cc(x):=\begin{pmatrix}
M_0(x) & M_1(x) & \cdots  & M_m(x) & 0 &  \cdots & 0\\
 & M_0(x) & \cdots & M_{m-1}(x) & M_m(x) &     &   \\
 & & \ddots &  & \ddots & \ddots & 0 \\
& & & M_0(x) & & \ddots & M_m(x) \\
& & & & M_0(x) & & M_{m-1}(x) \\
& & & &  & \ddots & \vdots \\
 & & & & & &  M_0(x) \\
\end{pmatrix},$$
where $M_i(x)$'s are the matrices defined as in \eqref{GrM} in which $\gamma$ has been replaced by the variable $x$. 
We extend this matrix to the matrix
$$\tilde{G}_L^\cc(x):=\begin{pmatrix}
M_0(x) & M_1(x) & \cdots  & M_m(x) & M_{m+1}(x) &  \cdots & M_{m+r}(x)\\
 & M_0(x) & \cdots & M_{m-1}(x) & M_m(x) &     &  \vdots \\
 & & \ddots &  & \ddots & \ddots & M_{m+1}(x) \\
& & & M_0(x) & & \ddots & M_m(x) \\
& & & & M_0(x) & & M_{m-1}(x) \\
& & & &  & \ddots & \vdots \\
 & & & & & &  M_0(x) \\
\end{pmatrix}.$$
where we have replaced the $0$ blocks in the topright part of $G_L^\cc(x)$ with the matrices $M_{m+i}(x)$.

Also in this case we choose $L+1$ integers $\bar\ell_0,\ldots, \bar\ell_{L}$ with the constraint that $\sum_{j=0}^s\bar\ell_j \leq (s+1)k$, for $s=0,\ldots, L-1$ and $\sum_{j=0}^L\bar\ell_j = (L+1)k$, and consider
a maximal submatrix of $G_L^\cc(x)$ obtained by selecting $\bar\ell_i$ columns from the $i$-th columns block, for $i=0,\ldots, L$, which we call $F(x)$. Moreover, let $\tilde{F}(x)$ be the corresponding $(L+1)k \times (L+1)k$ submatrix of $\tilde{G}_L^\cc(x)$. According to Lemma \ref{lem:MDPpolynomial}, we only need to prove that $\det(F(x)) \neq 0$, and then we can conclude as we did for Case \textbf{I-A}. 

First we observe the following relations between $f(x):=\det(F(x))$ and $\tilde{f}(x):=\det(\tilde{F}(x))$. We can write  $f(x)=\sum_{\sigma \in \mathcal S_{(L+1)k}}F^\sigma$ and  $\tilde{f}(x)=\sum_{\sigma \in \mathcal S_{(L+1)k}}\tilde{F}^\sigma$, where $F^\sigma=\prod_{i}F_{i,\sigma(i)}$ and $\tilde{F}^\sigma=\prod_i\tilde{F}_{i,\sigma(i)}$. We define
$\Theta_j:=\{\sigma \in \mathcal S_{(L+1)k}\mid \deg_x(F^\sigma)=j\}$ and $\tilde{\Theta}_j:=\{\sigma \in \mathcal S_{(L+1)k}\mid \deg_x(\tilde{F}^\sigma)=j\}$. We have
\begin{align*}f(x)&=\sum_{j} f_j x^j=\sum_{j}\Bigg(\sum_{\sigma \in \Theta_j} F^\sigma \Bigg) x^j,\\
\tilde{f}(x)&=\sum_{j} \tilde{f}_j x^j=\sum_{j}\Bigg(\sum_{\sigma \in \tilde{\Theta}_j} \tilde{F}^\sigma \Bigg) x^j.\end{align*}

By definition of ${G}_L^\cc(x)$ and $\tilde{G}_L^\cc(x)$, every entry of $G_L^\cc(x)$ is equal to the corresponding entry of  $\tilde{G}_L^\cc(x)$ or it is equal to $0$. Therefore, we get that $F^\sigma \in \{0, \tilde{F}^\sigma \}$, for every $\sigma \in \mathcal S_{(L+1)k}$. Hence we can also write
$$\tilde{f}(x)=\sum_{j} \tilde{f}_j x^j=\sum_{j}\left(\sum_{\sigma \in {\Theta}_j} \tilde{F}^\sigma \right) x^j.$$ 
In particular, if a monomial $\tilde{f}_tx^t$ of a certain degree $t$ in $\tilde{f}(x)$ is obtained from the matrix $\tilde{G}_L^\cc(x)$ without involving the blocks $G_{m+i}(x)$ for $i=1,\ldots, r$, then the monomial $f_tx^t$ of the same degree $t$  in $f(x)$ is the same, \ie $f_t=\tilde{f}_t$.

By using the same proof as in Case I-A, we can see that we can use Theorem \ref{thm:main} to show that $\tilde{f}(x)=\det(\tilde{F}(x))\neq 0$ and the monomial $\tilde{f}_Mx^M$ of minimum degree $M$ corresponds to the product of the $\bar\ell_i \times \bar\ell_i$ minors along the main diagonal (observe that here we allow $\bar\ell_i$ to be zero, by using the convention that the determinant of a $0\times 0$ matrix to be $1$).  If we  show that these minors do not involve any of the blocks $G_{m+j}(x)$ for $j=1,\ldots, r$, then we can deduce that $f_M=\tilde{f}_M \neq 0$, and this concludes the proof. Suppose by contradiction that one of the $\bar\ell_i \times \bar\ell_i$ minors involves one of the blocks $G_{m+j}(x)$ for $j=1,\ldots, r$. Then it must be $i=m+a$ for some $a=j,\ldots, r$. Moreover, this happens  if and only if 
$(m+r-a+j+1)k<\sum_{t=a}^r\bar\ell_{m+t}$. However, we have 
$\sum_{t=a}^r\bar\ell_{m+t} \leq n(r-a+1)$ and since $r=\lfloor\frac{km}{n-k} \rfloor$, also $ (n-k)r \leq  km$. Therefore, we get the following chain of inequalities 
\begin{align*}
n(r-a+j+1) &= nr+n(1-a+j)\leq km+kr+n(1-a+j) \\ 
&\leq km+kr+ k(1-a+j)
=k(m+r-a+j+1) \\
&<\sum_{t=a}^r\bar\ell_{m+t} \leq n(r-a+1),
\end{align*}
which yields  a contradiction. 

\underline{\textbf{Case II: $k \nmid \delta$}}.
In this case, $G_m=N_{m,t}$, where $1\leq t=\delta-(m-1)k\leq k-1$. At this point, we use a similar argument as done in Case \textbf{I-B}. Let $L=m-1+\lfloor\frac{k(m-1)+t}{n-k} \rfloor=m-1+r$, for some $r\geq 0$. Consider the polynomial matrix $G_L^\cc(x)$.
If $r=0$, then the matrix $N_{m,t}(x)$ does not appear in $G_L^\cc(x)$ and we conclude as in Case \textbf{I-A}. 
Therefore, assume $r\geq 1$. Observe that the matrix $N_{m,t}(x)$ is $0$ in the first $k-t$ rows, and it coincides with the matrix $M_m(x)$ in the last $t$ rows, We construct the matrix $\tilde{G}_L^\cc(x)$ from $G_L^\cc(x)$ by replacing the blocks $N_{m,t}(x)$ by $M_m(x)$ and the topright $0$ blocks by matrices $M_{m+i}(x)$, for $i=1,\ldots, r$, obtaining 
$$\tilde{G}_L^\cc(x):=\begin{pmatrix}
M_0(x) & M_1(x) & \cdots  & M_m(x) & M_{m+1}(x) &  \cdots & M_{m+r-1}(x)\\
 & M_0(x) & \cdots & M_{m-1}(x) & M_m(x) &     &  \vdots \\
 & & \ddots &  & \ddots & \ddots & M_{m+1}(x) \\
& & & M_0(x) & & \ddots & M_m(x) \\
& & & & M_0(x) & & M_{m-1}(x) \\
& & & &  & \ddots & \vdots \\
 & & & & & &  M_0(x) \\
\end{pmatrix}.$$
Now, for the same argument used in Case \textbf{I-B}, we just need to prove that none of $\bar\ell_i \times \bar\ell_i$ minors across the main diagonal involves any of the blocks $M_{m+j}(x)$, for $j \geq 1$, and neither the first $k-t$ rows of $M_m(x)$. By contradiction, assume that this happens in the block $m+a$, for some $a=0,\ldots, r-1$. This is true if and only if $\bar\ell_{m+a}+\ldots+ \bar\ell_{m+r-1}>\delta+k(r-a)$. However, we have 
$\sum_{t=a}^{r-1}\bar\ell_{m+t} \leq n(r-a)$ and since $r=\lfloor\frac{\delta}{n-k} \rfloor$, also $ (n-k)r \leq  \delta$. Therefore, we get the following chain of inequalities 
\begin{align*}
n(r-a) &\leq \delta+kr-na  \leq \delta+k(r-a) \\
&<\sum_{t=a}^{r-1}\bar\ell_{m+t} \leq n(r-a),
\end{align*}
 which yields  a contradiction. 
\end{proof}

\begin{remark}
It is important to point out that the same proof does not work if we  try to use Theorem \ref{thm:maxdegree} instead of Theorem \ref{thm:main}. Indeed, if $k$ divides $\delta$, whenever we select $\ell_0,\ldots, \ell_\mm$ columns from the $\mm+1$ column blocks and we have that one of the $\ell_i$'s for $i=1,\ldots, \mm$ is strictly greater than $k$, then the $\ell_i \times \ell_i$ submatrix that we select is not a classical Vandermonde, hence it is not guaranteed that when we evaluate the variables $\Y$ in $\mathrm{A}$ we get a nonzero determinant. This happens essentially for every choice of $\ell_i$'s, except when we take $\ell_i=k$ for every $i=0,\ldots,\mm$, in which case the minor is clearly  nonzero, since it is the  product of determinants of classical Vandermonde matrices. In the case that  $k$ does not divide $\delta$ it is even more clear that we cannot use Theorem \ref{thm:maxdegree}, since the $\ell_\mm\times \ell_\mm$  submatrix  that we should select will we have some zero rows.
\end{remark}

\section{Field size for MDP WRS convolutional codes}\label{sec:fieldsize}

In Section \ref{Sec:Construction} we gave a construction of WRS convolutional codes, which we proved in Theorem \ref{thm:ConstructionMDP} to be MDP, under the assumption that the extension degree $s$ is larger than the value $D(k,n,\delta,\alpha)$. Our goal in this section is to give an estimate on the required  field size. In particular, we will prove that for constructing WRS convolutional codes that are MDP we need a field of size $q^s$, where $q$ is any prime power greater than $n$ and $s=\mathcal O(\delta^3)$. It is straightforward to observe that for our base field $\Fq$ from which we take the vector $\alpha$ we need that $q>n$. This is because we require that $\alpha$ has pairwise distinct nonzero elements. The only thing that we need to estimate is the magnitude of the degree extension $s$, or equivalently, of the integer $D(k,n,\delta,\alpha)$.

\begin{proposition}\label{prop:Dknma}
  Let $k,n,\delta$ be integers such that $0<k<n$, let $m:=\lceil \frac{\delta}{k}\rceil$ and let $\alpha \in (\Fq^*)^n$ be a vector of nonzero pairwise distinct elements. Then,
     $$D(k,n,\delta,\alpha) \leq (L-m+1)\binom{\delta}{2}+k^2\binom{m}{3}+\binom{k}{2}\binom{m}{2}. $$
\end{proposition}

\begin{proof}
We provide an upper bound on the value $D(k,n,\delta,\alpha)$ by estimating the maximum degree of the polynomials in $\mathcal{P}(k,n,\delta,\alpha)$. To this end, we take the obvious upper bound in which we consider in each row of $G_L^\cc(x)$ the maximum degree, and then we sum up all these values. We divide the matrix in row blocks, indexed by $0,1,\ldots, L$.  Observe that in the row block $L$ only the matrix $G_0$ appears, in the row block $(L-1)$ the highest degrees are given by the row degrees in $x$ of $G_1(x)$, and so on. This means that for each $i\in\{0,\ldots,m-1\}$, the row block  $(L-i)$ has highest row degrees in $x$ given by those of $G_i(x)$. These matrices correspond to the matrices $M_i(x)$ and it is straightforward to see that the sum of the row degrees in $x$ of the matrix $M_i(x)$ is
$$ w_i:=\sum_{j=0}^{k-1}\binom{i}{2}k+ij=\binom{i}{2}k^2+i\binom{k}{2}. $$
For the remaining blocks, the degrees are the same due to the structure of $G_L^\cc(x)$ in which all the matrices $G_i$'s appear. Hence, we only can consider $(L-m+1)$ times the sum of row degrees of the first block. Observe that $G_m(x)=N_{m,t}(x)$, where $t:=\delta-k(m-1)$. Hence, the row degrees of this first block will consist of the row degrees of the last $t$ rows of $N_{m,t}(x)$, and the first $k-t$ row degrees of $M_{m-1}(x)$. Summing this quantities we get that the sum of the row degrees in $x$ of the first block is 
\begin{align*}w_{m,t}&:=\sum_{j=0}^{t-1}\Big(\binom{m}{2}k+mi\Big) + \sum_{j=t}^{k-1}\Big(\binom{m-1}{2}+(m-1)j\Big) \\
& = \frac{1}{2}t\left(km(m-1)+(t-1)m \right) + \frac{1}{2}(k-t)\left( k(m-1)(m-2)+(k+t-1)(m-1) \right) \\ 
&=  \frac{1}{2}\left(t(\delta-1)m+(k-t)(\delta-1)(m-1)\right) \\
& = \frac{1}{2}(\delta-1)\left(k(m-1)+t\right) \\
&= \binom{\delta}{2}. \end{align*}
Putting together all these quantities, we get
\begin{align*}
    D(k,n,\delta,\alpha) &\leq (L-m+1)w_{m,t}+\sum_{i=0}^{m-1}w_i \\
    & =(L-m+1)\binom{\delta}{2}+k^2 \sum_{i=0}^{m-1}\binom{i}{2}+\binom{k}{2}\sum_{i=0}^{m-1}i \\
    &=(L-m+1)\binom{\delta}{2}+k^2\binom{m}{3}+\binom{k}{2}\binom{m}{2}.
\end{align*}

\end{proof}

\bigskip

Next, we show an example considering our WRS convolutional codes where we compute the value of $D(k,n,\delta,\alpha)$ and compare it with the upper bound of Proposition \ref{prop:Dknma}. 
\begin{example}
Suppose that we want to construct a WRS convolutional code with parameters $k=2$, $n=7$ and $\delta=4$. We have that its memory $m$ has to be $2$ and also $L$ is equal to $2$. Moreover, we need to start with a base field $\Fq$ whose cardinality is $q>7$. In this example we consider the case $q=11$, that is $\F_{11}=\{0,1,\ldots, 10\}$.
We now fix the vector $\alpha\in \F_{11}^7$ to have nonzero pairwise distinct entries as $\alpha:=(1,2,3,4,5,6,7)$. In order to determine a suitable $\gamma$ for constructing the code $\C_{2,7}^4(\gamma, \alpha)$, we first want to compute the value $D(2,7,4,\alpha)$. For this purpose, we are now going to construct the matrix $G_2^\cc(x)$.

\begin{align*}
G_0(x) & =
\begin{pmatrix}
1 & 2 & 3 & 4 & 5 & 6 & 7 \\
1 & 1 & 1 & 1 & 1 & 1 & 1 \\
\end{pmatrix},\\
G_1(x) & =
\begin{pmatrix}
x  &  8 x &  5x &  9x &  4x &  7x &  2x \\
1 & 4 & 9 & 5 & 3 & 3 & 5
\end{pmatrix},\\
G_2(x) & =
\begin{pmatrix}
x^4 &  10x^4 & x^4 & x^4  & x^4  &   10 x^4 &   10 x^4 \\
x^2  &  5 x^2 &  4 x^2  & 3x^2 & 9x^2 & 9x^2 &  3x^2
\end{pmatrix}.
\end{align*}

The  $L$-th truncated sliding matrix $G_L^\cc(x)$ is
\begin{align*}
G_2^\cc(x) & =
\begin{pmatrix}
G_0(x) & G_1(x) & G_2(x) \\
    & G_0(x) & G_1(x)  \\
    &     & G_0(x)  \\
    \end{pmatrix}\in \F_{11}[x]^{6\times 21}.\\
\end{align*}
By Theorem \ref{lem:MDPpolynomial}, the matrix $G(z)$ generates a $\C_{2,7}^4(\gamma, \alpha)$ MDP convolutional code if it has the MDP property, \ie, the nontrivial full size minors of $G^\cc_2$ are all nonzero. A necessary field size for this to hold is $\F_{11^s}$ with $s > D(2,7,4,\alpha)$, according to Theorem \ref{thm:ConstructionMDP}. It can be checked that $D(2,7,4,\alpha)=4$, that comes from the full size minor $p(x)=7x^4+2x^2+7x+4 \in \mathcal P(2,7,4,\alpha)$ of $G_{2}^\cc(x)$ obtained selecting the columns with indices $\{1,8,9,15,16,17\}$. 
From Proposition~\ref{prop:Dknma}, we have that
$D(2,7,4,\alpha) \leq 7$, and  we can observe that the bound is bigger than the real value.
\end{example}

\bigskip

\begin{theorem}\label{thm:fieldsizeMDP}\; 
The WRS convolutional codes $\C_{k,n}^\delta(\gamma, \alpha)$ provide a family of MDP convolutional code over a field of size $q^s$, where $q>n$ and 
$$s> \frac{\delta^3}{2}\left(\frac{1}{(n-k)}+\frac{1}{3k}\right)+\frac{\delta^2}{2}\left(\frac{3}{2}-\frac{1}{(n-k)}-\frac{1}{2k}\right)+\delta\left(\frac{k}{12}-\frac{3}{4}\right).$$
\end{theorem}

\begin{proof} 
By Theorem \ref{thm:ConstructionMDP}, we get that the WRS convolutional code $\C_{k,n}^\delta( \gamma, \alpha)$ is MDP whenever $\alpha$ is a vector of pairwise distinct elements of $\Fq^*$, and $\gamma$ is a root of an irreducible polynomial in $\Fq[x]$ of degree $s>D(k,n,\delta,\alpha)$. Then,  we need $q>n$ and the result follows from an involved estimates of the bound on $D(k,n,\delta,\alpha)$ given in Proposition \ref{prop:Dknma}, which produce
$$ D(k,n,\delta,\alpha)\leq \frac{\delta^3}{2}\left(\frac{1}{(n-k)}+\frac{1}{3k}\right)+\frac{\delta^2}{2}\left(\frac{3}{2}-\frac{1}{(n-k)}-\frac{1}{2k}\right)+\delta\left(\frac{k}{12}-\frac{3}{4}\right). $$
\end{proof}

\begin{corollary}\label{cor:FinalFieldSize}
For every $k,n,\delta$ positive integers with $0<k<n-k$, there exists an $(n,k,\delta)_{q^s}$ MDP convolutional code where $q$ is any prime power greater than $n$ and 
$$s>\frac{\delta^3}{3t}+\frac{3\delta^2}{4}+\frac{\delta k}{12}.$$
\end{corollary}
\begin{proof}
  By Theorem  \ref{thm:fieldsizeMDP} we can construct an $(n,k,\delta)_{q^s}$ MDP WRS convolutional code $\C_{k,n}^\delta(\gamma,\alpha)$, where $q>n$ and $s>\frac{\delta^3}{2}\left(\frac{1}{(n-k)}+\frac{1}{3k}\right)+\frac{3\delta^2}{4}+\frac{\delta k}{4}$. Since $1/(n-k)\leq 1/k$, we obtain that for $s>\frac{\delta^3}{3k}+\frac{3\delta^2}{4}+\frac{\delta k}{4}$ we can construct an MDP WRS convolutional code with the desired parameters.
\end{proof}

The above results only provide upper bounds on the values $D(k,n,\delta,\alpha)$. In Table~\ref{BoundsD} we depict the actual values for some small parameters of $k,n$ and $\delta$ (found by exhaustive computer search) and compare them with the bound given in Proposition \ref{prop:Dknma}.

\begin{table}[!ht]
\begin{center}
\begin{tabular}{|c|c|c|} \hline
$[n,k,\delta]$ & $D(k,n,\delta,\alpha)$ bound  &   $D(k,n,\delta,\alpha)$  \\ \hline \hline
[2,1,1] & 0 &   0  \\ \hline
[2,1,2] & 3 &  2   \\ \hline
[3,2,2] & 3 & 3 \\ \hline
[3,1,2] & 2 &  1   \\ \hline
[3,2,1] & 0 &  0 \\ \hline
[4,2,2] & 2 &  2\\ \hline
[4,1,3] & 7 &  4 \\ \hline
[5,2,2] & 1 &  1 \\ \hline
[6,2,2] & 1 &  1\\ \hline
[6,2,3] & 1 &  1\\ \hline
[7,2,2] & 1 &  1 \\ \hline
[7,3,3] & 3 & 2  \\ \hline
\end{tabular}
\vspace{.5cm}
\caption{Comparison between the values $D(k,n,\delta,\alpha)$ obtained by computer search and the corresponding bounds given in Proposition~\ref{prop:Dknma}, for small parameters $n,k,\delta$.}
\label{BoundsD}
\end{center}
\end{table}

\subsection{Field size for memory 1}

In Proposition \ref{prop:Dknma} we gave a general upper bound on the value $D(k,n,\delta,\alpha)$, which then produces a sufficient field size for constructing WRS convolutional codes. 
 Here we improve this  upper bound  when the resulting code has memory $m=1$, that is when $\delta \leq k$. In order to do so, we will distinguish two cases: when $k<n-k$ and when $k=n-k$. 
 
For the first case, we will assume that $k<n-k$, that is when the WRS convolutional code has memory $1$. We first notice that  when $1\leq\delta <k <n-k$, we have $L=0$, and hence $D(k,n,\delta,\alpha)=0$. Therefore, the only interesting case is when $\delta=k$. 

\begin{proposition}\label{prop:case_memory1}
Suppose that $k<n-k$, Then 
$$D(k,n,k,\alpha)\leq \frac{k^2}{4}.$$
\end{proposition}

\begin{proof}
Since we are in the case $k<n-k$,  we have $L=m=1$, and hence the matrix $G_L^\cc(x)$ is of the form
$$G_1^\cc(x)= \begin{pmatrix} G_0 & G_1(x)  \\ 0 & G_0 \end{pmatrix}$$
For every $2k \times 2k$ submatrix of $G_L^\cc(x)$ that we need to consider, we can choose $k-i$ columns from the first block, and $k+i$ columns from the second block, for some $i=0,\ldots, k$. Let us fix the value $i$, and two subsets $I_1, I_2 \subseteq [n]$ such that $|I_1|=k-i, |I_2|=k+i$. Let $F(x)$ be the submatrix obtained by selecting the columns indexed by $I_1$ from the first block, and the columns indexed by $I_2$ from the second block, and let  $f(x):=\det(F(x))$. Then
$$F(x)= \begin{pmatrix} A_0 & B_1(x) \\ 0 & B_0 \\ \end{pmatrix},$$
where $A_0 \in \Fq^{k \times (k-i)}, B_0 \in \Fq^{k \times (k+i)}, B_1(x) \in \Fq[x]^{k \times (k+i)}$.
Computing $f(x)$ using Laplace formula on the first $k-i$ columns, we observe that the determinant will always be given by sum of products of a $(k-i) \times (k-i)$ minor of $A_0$ times the determinant of the remaining $(k+i) \times (k+i)$ minor in the second column block. Hence, the degree of $f(x)$ is at most the sum of the degrees of the first $i$ rows of $B_1(x)$, that is given by
$$\sum_{j=k-i}^{k-1}j.$$
By Theorem \ref{thm:main}, the minimum degree monomial of $f(x)$ is given by the product of the top $(k-i) \times (k-i)$ minor of $A_0$ times the determinant of the remaining submatrix of $\begin{pmatrix} B_1(x) \\ B_0 \end{pmatrix}$. The degree of this monomial is exactly $\sum_{j=0}^{i-1} j$, and thus
\begin{align*}
\deg(f(x))- \nu (f(x)) & \leq \sum_{j=k-i}^{k-1}j - \sum_{j=0}^{i-1} j  = \sum_{j=0}^{i-1} (k-i+j-j)  = i(k-i).
\end{align*}
Therefore, we get
\begin{align*}D(k,n,k,\alpha) & = \max\{ \deg f(x)-\nu (f(x)) \mid 0 \neq  f(x) \in \mathcal P(k,n,\delta,\alpha) \} \\
 & \leq \max\{ i(k-i) \mid i =0,\ldots,k\} \\
 & \leq \frac{k^2}{4}.
\end{align*}
\end{proof}

Notice that in the above analysis we only left out the case when $n=2k$. Also in this case, if $\delta<k$, we have $L=0$ and hence $D(k,2k,\delta,\alpha)=0$. Therefore, the only case left to study is $D(k,2k,k,\alpha)$, which we do in the following proposition. 

\begin{proposition}\label{prop:case_memory1_bis} For every positive integer $k$, we have
$$D(k,2k,k,\alpha)\leq \frac{k^2}{2}.$$
\end{proposition}

\begin{proof}
 In this case we have $m=1$ and $L=2$. Hence, the polynomial version of the sliding generator matrix is of the form
 $$ G_2^\cc(x)=\begin{pmatrix} G_0 & G_1(x) & 0 \\
 0 & G_0 & G_1(x) \\
 0 & 0 & G_0\end{pmatrix}.$$
 By the restriction of \eqref{eq:lis}, the only admissible full size submatrices we should consider are those obtained from $G_2^\cc(x)$ by selecting $\ell_0=k-i$ columns from the first columns block, $\ell_1=k-j+i$ columns from the second columns block and  $\ell_2=k+j$ columns from the last columns block, for any $i,j$ such that $0\leq i,j\leq k$. For a submatrix made with this choice, that we indicate as
 $$ F(x)=\begin{pmatrix}A_0 & B_1(x) & 0 \\
 0 & B_0 & C_1(x) \\
 0 & 0 & C_0\end{pmatrix},$$
 we can compute the minimum degree in $x$ of the resulting minor $f(x)=\det(F(x))$, that is $\nu(f(x))$. Even though there is a $0$-block in the top right corner, we can still use Theorem \ref{thm:main}, by using the same argument used in the proof of Theorem \ref{thm:ConstructionMDP}, and obtain that $\nu(f(x))=\sum_{t=0}^{i-1}t+\sum_{t=0}^{j-1}t=\binom{i}{2}+\binom{j}{2}$. The degree of $f(x)$ is going to be obtained by selecting the product of the  minors of the following three submatrices:
 the $\ell_0\times \ell_0$ submatrix obtained selecting the last $\ell_0$ of $A_0$; the $\ell_1\times \ell_1$ submatrix obtained by selecting the first $i$ rows of $B_1(x)$ and the last $k-j$ rows of $B_0$; the $\ell_2\times \ell_2$ submatrix obtained by selecting the first $j$ rows of $C_1(x)$ together with the whole matrix $C_0$. This gives the maximum possible degree, but since for some choice of the vector $\alpha$ this could be $0$, we have 
 $$\deg(f(x))\leq  \sum_{t=k-i}^{k-1}t+\sum_{t=k-j}^{k-1}t,$$ and therefore, 
 $$ \deg(f(x))-\nu(f(x))\leq \sum_{t=k-i}^{k-1}t+\sum_{t=k-j}^{k-1}t-\binom{i}{2}-\binom{j}{2}=i(k-i)+j(k-j).$$
 Thus, 
\begin{align*}D(k,2k,k,\alpha) & = \max\{ \deg f(x)-\nu (f(x)) \mid 0 \neq  f(x) \in \mathcal P(k,n,\delta,\alpha) \} \\
 & \leq \max\{ i(k-i) +j(k-j) \mid i,j =0,\ldots,k\} \\
 & \leq \frac{k^2}{2}.
\end{align*}
\end{proof}

\begin{remark}
 Proposition \ref{prop:case_memory1} and Proposition \ref{prop:case_memory1_bis} improve on the more general bound given in Proposition \ref{prop:Dknma}. Indeed,  with the latter result we obtain $D(k,n,k,\alpha)\leq \frac{k(k-1)}{2}$ when $k<n-k$, and $D(k,2k,k,\alpha) \leq k(k-1)$. Therefore, in both cases we refine the estimate by a factor $1/2$.
\end{remark}

\subsection{Comparison}

In this subsection we present  a comparison of the existing fields sizes required to build MDP convolutional codes for several sets of given parameters $(n,k,\delta)$. The  compared results are of different nature and need to be distinguished. Some of them are general bounds on the field size but with no associated concrete construction achieving such a bound. Others are conjectures or examples found by computer search. These differences are explained and analyzed in this subsection. 

As mentioned before, superregular matrices have been one of the fundamental tools for constructing MDP codes and the required field size to construct the codes has been often given in terms of the field size needed to build the associated superregular Toeplitz matrices. An upper triangular Toeplitz matrix 
\begin{align}\label{SR}
A=\left(\begin{array}{cccc}
    a_{0} &  \cdots & a_{r-1} & a_{r} \\
    & a_0 & \cdots & a_{r-1}\\
     &   & \ddots & \vdots \\
    &         &        & a_{0}\\
\end{array}\right)\end{align}
is superregular  if $a_i\neq 0$ for each $i=0, 1, \dots, r$ and all the square submatrices of $A$ with no zeros in the diagonal are nonsingular. It can be verified \cite{NaRo2015} that in order to built an $(n,k,\delta)$ MDP convolutional code an upper triangular Toeplitz superregular matrix, $A$, of size greater or equal than
\begin{equation}\label{eq:SRtoMDP}
r:= \max \{n-k,k\}(L +1)+\min \{n-k,k\}-1,
\end{equation}
needs to be constructed.

In \cite{AlmeidaNappPinto2013} and \cite{gl03} two general classes of superregular matrices of any size were presented. The lower bound on the field size required to build the superregular lower triangular Toeplitz matrix $A\in\F^{r\times r}$ in \cite{gl03} is  $|\F|>c^r r^{r/2}$ where $c= \binom{r-1}{\lfloor \frac{r-1}{2}\rfloor}$. 
For the one provided in \cite{AlmeidaNappPinto2013} the lower bound is given by $|\F| \geq 2^{(2^{(r+2)})}$. For upper bounds on the size of a field $\F$  to ensure the existence (without providing a concrete construction) of a superregular lower triangular Toeplitz matrix over $\F$, see \cite{Hutchinson2008} and \cite{Lieb2019}. Based on examples derived by computer search, it was conjectured in \cite[Conjecture 3.5]{Hutchinson2008} and \cite{gl03} that for $r\geq 5$ there exists a superregular lower triangular Toeplitz matrix of order $r$ over the field
$\mathbb F_{2^{r-2}}$.

We compare these results in Table \ref{texam} together with some examples in \cite{AlNa2020} and \cite{ma16} found by optimized computer search. 

\begin{remark}\label{rem:asymptotic}
 Here we compare asymptotically the field size needed for our  WRS convolutional codes to be MDP with the other two existing general constructions of MDP convolutional provided in \cite{AlmeidaNappPinto2013,gl03}. 
We consider the case in which we fix the rate of the code to be constant $R:=\frac{k}{n}$, and express all the field sizes in terms of $R,\delta$ and $n$. First, notice that the parameter $r$ defined above can be approximated by $ R^{-1}\delta+n-1$. Assuming now that the value $R^{-1}\delta+n-1$ grows, we study the asymptotic behaviours of the field sizes.

By using Stirling approximation formula, we have that 
$$c:=\binom{r-1}{\lfloor \frac{r-1}{2}\rfloor}\sim \frac{2^{R^{-1}\delta+n}}{\sqrt{8\pi(R^{-1}\delta+n-2)}},$$ leading to the following  asymptotic approximation for the field size needed in \cite{gl03}:
\begin{equation}\label{eq:fieldsizeGl03} |\F|\sim c^rr^{r/2}\sim e \cdot \frac{2^{(R^{-1}\delta +n-1)(R^{-1}\delta +n-\frac{3}{2})}}{\sqrt{\pi}^{(R^{-1}\delta+n-1)}}\sim  2^{(R^{-1}\delta +n-1)(R^{-1}\delta +n-\frac{3}{2}-\log_2(\pi))}.
\end{equation}
Moreover, the field size needed for the construction in \cite{AlmeidaNappPinto2013} is
$$ |\F| \sim 2^{2^{R^{-1}\delta+n+1}},$$
which is always asymptotically worse than \eqref{eq:fieldsizeGl03}. 
Finally, by Corollary \ref{cor:FinalFieldSize}, we have that for our constructions we need a field size of approximately
\begin{equation}\label{eq:asymptotic_fieldsize} |\F| \sim 2^{\log_2(n)\frac{11}{24}(\frac{\delta^3}{Rn}+\delta R n)}.
\end{equation}

In order to compare \eqref{eq:asymptotic_fieldsize} with \eqref{eq:fieldsizeGl03}, we compare the asymptotics of their logarithms. If $\delta$ is constant,  then our construction is better, while if $n$ is constant the one of \cite{gl03} is better. Suppose now that none of $\delta$ and $n$ is constant. We have that whenever $\delta =\Theta(n^{1-\epsilon}(\log n)^\beta)$ for any $0<\epsilon < 1$ and $\beta \in \mathbb R$, our construction beats the field size of \cite{gl03}, while when $\delta=\Theta(n^{1+\epsilon}(\log n)^\beta)$ for any $\epsilon >0$ and $\beta \in \mathbb R$, the field size of \cite{gl03} is asymptotically better than ours. Furthermore, when $\delta=\Theta(n(\log n)^{-1-\beta})$  for any $\beta >0$  our field size is better, while in the case $\delta=\Theta(n(\log n)^{-1+\beta})$ for any $\beta >0$, the one in \cite{gl03} is smaller than ours. Finally, in the case that $\delta=\Theta(n(\log n)^{-1})$, the logarithms of the field sizes are asymptotically the same, so one should carefully investigate the smallest order terms.
\end{remark}

\begin{table}[!ht]
\begin{center}
\tiny{\begin{tabular}{|c|c|c|c|c|c|c|c||c|} \hline 
$[n,k,\delta]$ &  &  & & &   &  &   &  \\
L, m, $\mu$, $r$  &   \multirow{-2}{*}{ \cite{AlmeidaNappPinto2013}  }    & \multirow{-2}{*}{\cite{gl03} }& \multirow{-2}{*}{\cite{ma16}$^\ast$} & \multirow{-2}{*}{\cite{AlNa2020}$^\ast$}  &  \multirow{-2}{*}{\cite{Hutchinson2008}$^{\dagger}$}   & \multirow{-2}{*}{\cite{Lieb2019}$^{\dagger}$}  &  \multirow{-2}{*}{$\C_{k,n}^\delta$} &  \multirow{-2}{*}{ \cite{gl03}$^\ddagger$}   \\ \hline \hline
$\begin{array}{c}
                              [2,1,1] \\
                               2,1,1,3
                            \end{array}$ & ${2^8}$ & $43$ &  ${2^5}$ & \cellgreen $3$ & \cellgreen $3$  & $55$  & \cellgreen $3$ &  -- \\  \hline 
$\begin{array}{c}
                              [2,1,2] \\
                               4,2,2,5
                            \end{array}$  & ${2^{32}}$  & $434692$ & ${2^7}$ & \cellgreen $7$ & $11$  & $1261$ &  ${27}$  & ${2^{3}}$ \\ \hline
$\begin{array}{c}
                              [3,2,2] \\
                               3,1,2,8
                            \end{array}$  & ${2^{512}}$  & $ 5^8 7^8 2^{12} +1$ & ${2^{11}}$ & \cellgreen ${31}$ & $233$  & $1981$ &${{256}}$  & ${2^6}$ \\ \hline
$\begin{array}{c}
                              [3,1,2] \\
                               3,2,1,8
                            \end{array}$ & ${2^{512}}$  & $5^8 7^8 2^{12} +1$ & ${2^{11}}$ & ${31}$ & $233$  & $3961$ & \cellgreen $16$  & ${2^6}$\\ \hline
$\begin{array}{c}
                              [3,2,1] \\
                               1,1,1,4
                            \end{array}$ & ${2^{32}}$  & $2^4 3^4 +1$  & -- & $5$ & $5$   & \cellgreen ${3}$   & $4$ & ${2^2}$ \\ \hline
$\begin{array}{c}
                              [4,2,2] \\
                               2,1,1,7
                            \end{array}$ & ${2^{128}}$  & $\sim 10^{12}$ & -- & \cellgreen ${17}$ & $77$  & $5545$ &  $125$ & ${2^{5}}$ \\ \hline
$\begin{array}{c}
                              [4,1,3] \\
                               4,3,1,15
                            \end{array}$ & ${2^{2^{17}}}$  & $ \sim 7\cdot10^{61} $ & -- & -- & $1338936$ & $232561$  &  \cellgreen ${3125}$  &  ${2^{13}}$ \\ \hline
$\begin{array}{c}
                              [5,2,2] \\
                               1,1,1,7
                            \end{array}$ & $2^{2^9}$ & $\sim 10^{12}$ & --  & \cellgreen $17$  & $77$ & $35$ &  $49$ &  $32$ \\ \hline
$\begin{array}{c}
                              [6,2,2] \\
                               1,1,1,9
                            \end{array}$ & $2^{2^{11}}$ & $\sim 7\cdot 10^{20}$ &  -- & $59$ & $751$ & $71$ & \cellgreen $49$  & $128$ \\ \hline
$\begin{array}{c}
                              [6,2,2] \\
                               1,2,1,9
                            \end{array}$ & $2^{2^{11}}$ & $\sim 7\cdot 10^{20}$ &  -- & $59$ & $751$ & $71$ & \cellgreen $49$ &  $128$ \\ \hline
                            $\begin{array}{c}
                              [7,2,2] \\
                               1,1,1,11
                            \end{array}$ & $2^{2^{13}}$ & $\sim 10^{32}$ &  -- & -- & $8525$ & $126$ &   \cellgreen $64$ &   $512$ \\ \hline
$\begin{array}{c}
                              [7,3,3] \\
                               1,1,1,10
                            \end{array}$ & $2^{2^{12}}$ & $\sim 10^{26}$ &  -- & \cellgreen $127$ & $2495$ & $532$ &   $512$ &    $256$ \\ \hline
\end{tabular}}
\vspace{.5cm}

\caption{Parameters and smallest field sizes of MDP convolutional codes, according to known results in the literature.
 The columns are marked with $^{*}$ if the result is found by computer search; the results marked with $^{\dagger}$ indicates that they are not constructive; the symbol $^{\ddagger}$ means that the correspondent result is based on a conjecture.  
 The symbol -- indicates that there are no constructions for such parameters. For the nonconstructive results marked with $^\dagger$ we included the smallest field size needed, even if it is not a prime power.
The cells with the colored background indicate the best field size for the given parameters}\label{texam}
\end{center}
\end{table}

\section{Conclusions and future work}\label{sec:conclusions}

In this work we provided a new algebraic construction of convolutional codes for any choice of length ($n$), dimension ($k$) and degree ($\delta$),  which  are built upon generalized Vandermonde matrices. We called them weighted Reed-Solomon (WRS) convolutional codes and showed that, under some constraints on the field size, they are maximum distance profile (MDP). The proof of this result requires a technical study of multivariate polynomials and it is quite involved. A careful analysis of the field size shows that our construction requires a smaller field size compared to all the other constructions for some regime of the parameters. In Remark \ref{rem:asymptotic}, we show that our field size is  always smaller than the one in \cite{AlmeidaNappPinto2013}, and  comparable to the field size needed for the construction presented in \cite{gl03}, depending on the asymptotic behaviour of $n/\delta$. 
The algebraic structure of WRS convolutional codes inherited by  Reed-Solomon block codes opens many perspectives and has still to be explored, giving avenue for future research. In particular, we leave for the future the nontrivial task to derive an optimal algebraic decoding algorithm for this class of convolutional codes. Also, it is left as an open problem to show whether our construction is noncatastrophic, \ie, $G(z)$ is basic in the case in which $k$ does not divide $\delta$. 

\bigskip
\bigskip
\section*{Acknowledgements}
This work was partially supported by the Swiss National Science Foundations through grants no. 188430 and 187711 and by the Spanish Ministerio de Ciencia e Innovación via the grant with ref. PID2019-108668GB-I00. 

\bigskip
\bigskip
\bibliographystyle{plain}
\bibliography{biblio}

\begin{thebibliography}{10}

\bibitem{alfarano2020simplified}
G.~N. Alfarano and J.~Lieb.
\newblock On the left primeness of some polynomial matrices with applications
  to convolutional codes.
\newblock {\em Journal of Algebra and Its Applications}, 2020.

\bibitem{AlmeidaNappPinto2013}
P.~Almeida, D.~Napp, and R.~Pinto.
\newblock A new class of superregular matrices and {MDP} convolutional codes.
\newblock {\em Linear Algebra and its Applications}, 439(7):2145--2157, 2013.

\bibitem{al16}
P.~Almeida, D.~Napp, and R.~Pinto.
\newblock Superregular matrices and applications to convolutional codes.
\newblock {\em Linear Algebra and its Applications}, 499:1--25, 2016.

\bibitem{AlmeidaLieb}
P.~J. {Almeida} and J.~{Lieb}.
\newblock Complete j-{MDP} convolutional codes.
\newblock {\em IEEE Transactions on Information Theory}, 66(12):7348--7359,
  2020.

\bibitem{AlNa2020}
P.~J. {Almeida} and D.~{Napp}.
\newblock Superregular matrices over small finite fields.
\newblock {\em arXiv preprint arXiv:2008.00215}, 2020.

\bibitem{ba15b}
A.~Badr, A.~Khisti, Wai-Tian. Tan, and J.~Apostolopoulos.
\newblock Layered constructions for low-delay streaming codes.
\newblock {\em {IEEE} Trans. Inform. Theory}, 63(1):111--141, 2017.

\bibitem{cl12}
J.J. Climent, D.~Napp, C.~Perea, and R.~Pinto.
\newblock A construction of {MDS} {$2$D} convolutional codes of rate $1/n$
  based on superregular matrices.
\newblock {\em Linear Algebra and its Applications}, 437:766--780, 2012.

\bibitem{cl16}
J.J. Climent, D.~Napp, C.~Perea, and R.~Pinto.
\newblock Maximum distance separable {$2$D} convolutional codes.
\newblock {\em IEEE Transactions on Information Theory}, 62(2):669--680, 2016.

\bibitem{Esmaeili99}
M.~{Esmaeili}, T.~A. {Gulliver}, N.~P. {Secord}, and S.~A. {Mahmoud}.
\newblock A link between quasi-cyclic codes and convolutional codes.
\newblock {\em IEEE Transactions on Information Theory}, 44(1):431--435, 1998.

\bibitem{Khisti19}
S.~L. {Fong}, A.~{Khisti}, B.~{Li}, W.~{Tan}, X.~{Zhu}, and
  J.~{Apostolopoulos}.
\newblock Optimal streaming codes for channels with burst and arbitrary
  erasures.
\newblock {\em IEEE Transactions on Information Theory}, 65(7):4274--4292,
  2019.

\bibitem{forney75}
G.D. {Forney, Jr.}
\newblock Minimal bases of rational vector spaces, with applications to
  multivariable linear systems.
\newblock {\em SIAM Journal on Control}, 13:493--520, 1975.

\bibitem{Gluesing06}
H.~Gluesing-Luerssen and B.~Langfeld.
\newblock A class of one-dimensional {MDS} convolutional codes.
\newblock {\em Journal of Algebra and Its Applications}, 05(04):505--520, 2006.

\bibitem{gl03}
H.~Gluesing-Luerssen, J.~Rosenthal, and R.~Smarandache.
\newblock Strongly {MDS} convolutional codes.
\newblock {\em IEEE Transactions on Information Theory}, 52(2):584--598, 2006.

\bibitem{LAGUARDIA2014}
G.~G.~La Guardia.
\newblock On negacyclic {MDS}-convolutional codes.
\newblock {\em Linear Algebra and its Applications}, 448(Supplement C):85 --
  96, 2014.

\bibitem{HaOs2018}
J.~Hansen, J.~{\O}stergaard, J.~Kudahl, and J.~H. Madsen.
\newblock Superregular lower triangular {T}oeplitz matrices for low delay
  wireless streaming.
\newblock {\em IEEE Transactions on Communications}, 65(9):4027--4038, 2017.

\bibitem{Hutchinson2005}
R.~Hutchinson, J.~Rosenthal, and R.~Smarandache.
\newblock Convolutional codes with maximum distance profile.
\newblock {\em Systems \& Control Letters}, 51(1):53--63, 2005.

\bibitem{Hutchinson2008}
R.~Hutchinson, R.~Smarandache, and J.~Trumpf.
\newblock On superregular matrices and {MDP} convolutional codes.
\newblock {\em {Linear Algebra and its Applications}}, 428:2585--2596, 2008.

\bibitem{jo99}
R.~Johannesson and K.~Sh. Zigangirov.
\newblock {\em Fundamentals of Convolutional Coding}.
\newblock IEEE Press, New York, 2015.

\bibitem{Justesen1972}
J.~Justesen.
\newblock A class of constructive asymptotically good algebraic codes.
\newblock {\em IEEE Transactions on Information Theory}, 18(5):652--656, 1972.

\bibitem{Justesen1975}
J.~Justesen.
\newblock An algebraic construction of rate $1/\nu$ convolutional codes.
\newblock {\em IEEE Transactions on Information Theory}, 21(1):577--580, 1975.

\bibitem{ka80}
T.~Kailath.
\newblock {\em Linear systems}.
\newblock Prentice Hall information and system sciences series. Prentice-Hall,
  Englewood Cliffs, 1980.

\bibitem{Levy1993}
Y.~Levy and D.~J. Costello, Jr.
\newblock An algebraic approach to constructing convolutional codes from
  quasicyclic codes.
\newblock {\em DIMACS Series in Discrete Mathematics and Theoretical Computer
  Science}, 14:189--198, 1993.

\bibitem{Lieb2019}
J.~Lieb.
\newblock Necessary field size and probability for {MDP} and complete {MDP}
  convolutional codes.
\newblock {\em Designs, Codes and Cryptography}, 87(12):3019--3043, 2019.

\bibitem{LiebPinto}
J.~Lieb and R.~Pinto.
\newblock Constructions of {MDS} convolutional codes using superregular
  matrices.
\newblock {\em Journal of Algebra Combinatorics Discrete Structures and
  Applications}, 7(1), 2020.

\bibitem{ma77}
F.~J. MacWilliams and N.~J.A. Sloane.
\newblock {\em The Theory of Error-Correcting Codes}.
\newblock North Holland, Amsterdam, 1977.

\bibitem{ma16}
R.~Mahmood, A.~Badr, and A.~Khisti.
\newblock Convolutional codes with maximum column sum rank for network
  streaming.
\newblock {\em IEEE Transactions on Information Theory}, 62(6):3039--3052,
  2016.

\bibitem{Massey1973}
J.~L. Massey, D.~J. Costello, and J.~Justesen.
\newblock Polynomial weights and code constructions.
\newblock {\em IEEE Transactions on Information Theory}, 19(1):101--110, 1973.

\bibitem{NaRo2015}
D.~Napp and R.~Smarandache.
\newblock Constructing strongly {MDS} convolutional codes with maximum distance
  profile.
\newblock {\em Advances in Mathematics of Communications}, 10(2):275--290,
  2016.

\bibitem{Plaza13}
F.~J. {Plaza-Martín}, J.~I. {Iglesias-Curto}, and G.~{Serrano-Sotelo}.
\newblock On the construction of 1-{D} {MDS} convolutional {G}oppa codes.
\newblock {\em IEEE Transactions on Information Theory}, 59(7):4615--4625,
  2013.

\bibitem{ro99a1}
J.~Rosenthal and R.~Smarandache.
\newblock Maximum distance separable convolutional codes.
\newblock {\em Applicable Algebra in Engineering, Communication and Computing},
  10(1):15--32, 1999.

\bibitem{Smarandache2001}
R.~Smarandache, H.~Gluesing-Luerssen, and J.~Rosenthal.
\newblock Constructions of {MDS}-convolutional codes.
\newblock {\em IEEE Transactions on Information Theory}, 47(5):2045--2049,
  2001.

\bibitem{Tanner1987}
R.~M. Tanner.
\newblock Convolutional codes from quasicyclic codes: A link between the
  theories of block and convolutional codes.
\newblock Technical Report USC-CRL-87-21, University of California, Santa Cruz,
  CA, November 1987.

\bibitem{to12}
V.~Tomas, J.~Rosenthal, and R.~Smarandache.
\newblock Decoding of convolutional codes over the erasure channel.
\newblock {\em IEEE Transactions on Information Theory}, 58(1):90--108, January
  2012.

\end{thebibliography}
\end{document}